\documentclass[letterpaper,11pt]{article}
\usepackage{tabularx} 
\usepackage{amsmath}  
\usepackage{graphicx} 

\usepackage[margin=2.8cm,letterpaper]{geometry} 
\usepackage{xr-hyper}

\usepackage{hyperref} 
\hypersetup{
	colorlinks=true,       
	linkcolor=blue,        
	citecolor=blue,        
	filecolor=magenta,     
	urlcolor=blue         
}
\usepackage{setspace}

\makeatletter
\newcommand*{\addFileDependency}[1]{
  \typeout{(#1)}
  \@addtofilelist{#1}
  \IfFileExists{#1}{}{\typeout{No file #1.}}
}
\makeatother

\newcommand*{\myexternaldocument}[1]{%
    \externaldocument{#1}%
    \addFileDependency{#1.tex}%
    \addFileDependency{#1.aux}%
}
\myexternaldocument{supplementary}

\usepackage{amsthm}
\usepackage{amssymb}
\usepackage{multirow}
\usepackage{array,multirow}
\usepackage{booktabs}
\usepackage{subcaption}
\usepackage{mwe}
\usepackage[export]{adjustbox}
\usepackage{enumitem}
\usepackage{listings}
\usepackage{threeparttable}
\usepackage{romannum}
\usepackage{bm}
\usepackage{bbm}
\usepackage{rotating}
\usepackage[]{natbib}

\usepackage[skip=1ex]{caption}
\newcommand\numberthis{\addtocounter{equation}{1}\tag{\theequation}}
\usepackage[ruled,vlined]{algorithm2e}

\usepackage{float}
\usepackage[utf8]{inputenc}
\usepackage{textcomp}
\usepackage{tabularx,ragged2e,booktabs,caption}
\newcolumntype{C}[1]{>{\Centering}m{#1}}

\usepackage{geometry}
\usepackage{float}
\newtheorem{example}{Example}

\newtheorem{proposition}{Proposition}
\newtheorem{remark}{Remark}

\title{\textbf{Learning Large $Q$-matrix by Restricted Boltzmann Machines}}

\author{Chengcheng Li, Chenchen Ma and Gongjun Xu\\ Department of Statistics, University of Michigan\footnote{This research is partially supported by NSF CAREER SES-1846747, DMS-1712717, SES-1659328.}}
\date{}

\begin{document}
\pagenumbering{arabic}

\maketitle

\begin{abstract}
Estimation of the large $Q$-matrix in Cognitive Diagnosis Models (CDMs) with many items and latent attributes from observational data has been a huge challenge due to its high computational cost.
Borrowing ideas from deep learning literature, we propose to learn the large $Q$-matrix by Restricted Boltzmann Machines (RBMs) to overcome the computational difficulties. 
In this paper, key relationships between RBMs and CDMs are identified.
Consistent and robust learning of the $Q$-matrix in various CDMs is shown to be valid under certain conditions. 
Our simulation studies under different CDM settings show that RBMs not only outperform the existing methods in terms of learning speed, but also maintain good recovery accuracy of the $Q$-matrix. 
In the end, we illustrate the applicability and effectiveness of our method through  a TIMSS mathematics data set.\\

\end{abstract}

\newpage
\section{Introduction}
Cognitive Diagnosis Models (CDMs) are popular statistical tools widely applied to educational assessments and psychological diagnoses, 
which have been receiving increasingly more attention in the past two decades. 
In many modern assessment situations, examiners are concerned with specific attributes that the examinees possess, and thus a simple overall score is no longer sufficient to depict the whole picture of the candidates. 
As a result, a finer evaluation of the examinees' attributes is desired. 
CDMs are such tools. 
They model the relationship between the test items and the examinees’ latent skills, which is helpful in assessment design and post-assessment analysis of the examinees’ latent attribute patterns. 
CDMs have seen vast applications in multiple scientific disciplines, including educational assessments \citep{educationassessment1, education2, educationassessment3}, psychiatric diagnosis of mental disorders \citep{Templin2006, pyctest2}, epidemiological and medical measurement studies \citep{disease}. 

Many CDMs can be viewed as restricted latent class models that directly model the response probabilities as functions of discrete latent attributes.
A common goal of cognitive diagnoses is to learn the examinees' latent attributes, such as personalities or skills, based on their responses to a combination of specially designed test items. 
The $Q$-matrix plays a critical role in CDMs.
It specifies the dependency structure between the test items and the latent attributes.
Knowing the $Q$-matrix accurately is important because it is indispensable to cognitive diagnoses. 
Besides, the $Q$-matrix itself can be used to categorize the test items and enable efficient design of future assessments. 
However, in reality, many existing assessments do not even have the $Q$-matrix explicitly specified.
Even the assessment providers specify the $Q$-matrix when designing the assessment, the specification may still be inaccurate.
In many cases, one test item may potentially be linked to multiple attributes, but usually only the most direct and apparent ones are identified in the pre-designed $Q$-matrix. 
Therefore, it is of paramount importance to develop methodologies to efficiently learn the $Q$-matrix from the observational responses.

Various approaches have been proposed in the literature to learn the $Q$-matrix. 
Those methods can be generally classified into two categories, validation of the existing $Q$-matrix \citep{delatorre2008, DeCarlo2012, Chiu2013, delatorie2016} and direct estimation of the $Q$-matrix from the observational data \citep{Liu_Xu_Z2012, chen_liu_Xu_ying2015, XuEstimateQmatrix, Chung2018AnMA, Chen2018, 7232852c3cdb4a2ebcac33aa88c1fedc}. 
However, most of the existing estimation methods for the whole $Q$-matrix in general suffer from huge computational cost and are not scalable with the size of the $Q$-matrix;
they either break down or are extremely computationally expensive even when the $Q$-matrix is moderately large. 
The high computational cost stems from the large number of configurations of the $Q$-matrix. 
If we view each binary element of the $Q$-matrix as a unique parameter, then the number of different configurations would grow exponentially with the size of the $Q$-matrix. 
In many applications, the number of latent attributes being tested is large, leading to a high-dimensional space for all possible latent attribute patterns. 
It is not uncommon that the number of potential attribute patterns is large, sometimes even larger than the sample size, making the estimation even more difficult. 
Such examples can be found in many applications,
such as educational assessments \citep{largelatentAttributes1, largelatentAttributes2} and the medical diagnosis of disease etiology \citep{disease}; for instance,  Section \ref{sec-real data analysis} presents a dataset from the
Trends in International Mathematics and Science Study (TIMSS), which has 13 binary latent attributes and $2^{13} = 8192$ attribute patterns while only 757 examinees.
On the other hand, the number of items being tested may also be large in many applications.
One example is the TIMSS mathematical test which often have more than 100 test items.
Another  example is the ADM admissions test, which is given twice a year and is used as an entrance test to universities and colleges, contains a total of 200 items \citep{gonzalez2017applying}.
Therefore, it remains an open and challenging problem to learn the large  $Q$-matrix from the observational data.

Borrowing the idea from the deep learning literature, we propose to use the restricted Bolzmann machines (RBMs) to learn the large $Q$-matrix. 
An RBM is a generative two-layer neural network that can learn a probability distribution over a collection of inputs \citep{RBM}. 
Amongst these inputs, some are observed variables while the others are latent variables that we do not observe, which matches the restricted latent class CDM setting. 
The weight matrix $\boldsymbol{W}$ in RBMs determines the relationship between the observed variables and the latent variables. 
By learning this weight matrix $\boldsymbol{W}$ under the framework of RBMs, we show that the structure of the $Q$-matrix in CDMs can be inferred accordingly.
Although this is similar to the maximum likelihood learning approach, by tapping on RBMs, fast learning of the large $Q$-matrix can be achieved.

Our main contributions are that we identify the relationships between CDMs and RBMs, and proposed a new way of learning the large $Q$-matrix efficiently. 
As far as we know, our proposed method is among the first ones in the literature that is scalable with the size of the $Q$-matrix (with computational cost of $O(J \times K)$) while at the same time retains high estimation accuracy. 
For example, comparing to \cite{XuEstimateQmatrix} which attains an estimation accuracy of $71.2\%$ in the GDINA setting with five independent latent attributes using $2000$ observations, 
our method achieves 
more than $86\%$ overall accuracy and much faster computational speed. 
Another interesting finding is that 
learning of the $Q$-matrix by RBMs is robust to different CDMs, including the DINA, ACDM and GDINA models. 
We provide theoretical guarantees under certain conditions and conduct simulation studies to support our findings.  
Besides, because of the unsupervised learning nature of RBMs,
the traditional cross-validation (CV) procedure are not directly applicable. 
As a result, we also present a new CV procedure specifically to the $Q$-matrix learning setting. 

The remaining parts of the paper are organized as follows. Section \ref{sec-CDM&RBM} gives reviews on CDMs and RBMs, and discussion of their relationships and why the learning of the $Q$-matrix by RBMs is achieveable across different CDMs. 
Section \ref{sec-method&algo} introduces our proposed estimation method and the new CV procedure.
Section \ref{sec-simulation-studies} consists of simulation studies on data generated from three typical CDMs. 
Section \ref{sec-real data analysis} demonstrates the performance of our proposed method through the data analysis on a TIMSS mathematics data set. 
Section \ref{sec-conclusion} concludes with discussions and potential future directions.
All the proofs and additional simulation results can be found in the Supplementary Materials.

\section{Estimation of Q-matrix Using RBMs}
\label{sec-CDM&RBM}

\subsection{Review of CDMs}
\label{sec-review}

Many CDMs have been developed in recent decades, among which the Deterministic Input Noisy output “And” gate model  \citep[DINA,][]{haertel1989,educationassessment1} is one of the most popular and simple models and serves as the foundation for many complex CDMs. 
Other popularly used CDMs include the Noisy Input Deterministic “And” gate model   \citep[NIDA,][]{educationassessment1}, the Reduced Reparametrized Unified
Model \cite[R-RUM,][]{Hartz2002}, the
General Diagnostic Model  \citep[GDM,][]{von2005}, the Deterministic Input Noisy “Or” gate   \cite[DINO,][]{Templin2006},  the Log linear CDM \citep[LCDM,][]{LogCDM},  the Additive CDM \citep[ACDM,][]{GDINA} and the Generalized DINA model \cite[GDINA,][]{GDINA}. 

Consider a CDM with $J$ items and $K$ latent attributes.
There are two types of variables for each subject: the observed responses for $J$ items $\mathbf{R}=(R_{1},..., R_{J})$ and the latent attribute pattern $\bm{\alpha}=(\alpha_{1},..., \alpha_{K})$, which are both assumed to be binary. 
$R_{j} \in \{1,0\}$ denotes whether the examinee answers item $j$ correctly and $\alpha_{k} \in \{1,0\}$ denotes possession or non-possession of the attribute $k$.
The $Q$-matrix,  $\boldsymbol{Q} =(q_{j,k}) \in \{0,1\}^{J\times K}$, specifies the dependence structure between the items and the latent attributes; $q_{j,k} \in \{1, 0\}$ denotes whether a correct response to item $j$ requires the latent attribute $k$. 
If we denote the $j$th row of the $Q$-matrix to be $\mathbf{q}_{j}$, then $\mathbf{q}_{j}$ reflects the full attribute requirements of item $j$. 
For a latent attribute pattern $\bm{\alpha}$, we say $\bm{\alpha}$ possesses all the required attributes of item $j$ if $\bm{\alpha} \succeq \mathbf{q}_{j}$, where $\bm{\alpha} \succeq \mathbf{q}_{j}$ means $\alpha_{k} \geq q_{j,k}$ for all $k=1,...,K$. 
Different CDMs model the item response functions $P(R_j=1\mid \bm{\alpha})$ differently with the item parameters constrained by the $Q$-matrix and specific cognitive diagnostic assumptions.
Below we introduce three popular CDMs that will be considered in later discussions.

\begin{example}[DINA model] 
Let $R_{i,j} \in \{1,0\}$ denote whether the subject $i$ answers the item $j$ correctly. 
Under the DINA model \citep{haertel1989,educationassessment1}, for the $j$th item and the $i$th subject with the latent attribute pattern $\bm{\alpha}_i=(\alpha_{i,1},\ldots,\alpha_{i,K})$, the ideal response variable is defined as $\xi_{i,j}=\prod_{k:q_{j,k}=1}\alpha_{i,k}=\prod_{k=1}^{K}\alpha_{i,k}^{q_{j,k}}$. 
The ideal response $\xi_{i,j}=1$ only if $\boldsymbol{\alpha}_i \succeq \boldsymbol{q}_j$, that is, the subject $i$ needs to possess all the latent attributes required by the item $j$ to have a positive ideal response.
The uncertainty is further incorporated by two parameters: the slipping parameter $s_{j}$ and the guessing parameter $g_{j}$.
Specifically, $s_{j}=P(R_{i,j}=0\mid\xi_{i,j}=1)$ and $g_{j}=P(R_{i,j}=1\mid\xi_{i,j}=0)$.
The slipping parameter and the guessing parameter further satisfy $1-s>g$, which indicates that the capable subjects will have higher positive probability than the incapable ones.
The DINA model is one of the most restrictive and interpretable CDMs for dichotomously scored test items. 
It is a parsimonious model that requires only two parameters for each item regardless of the number of attributes required for the item. 
It is appropriate when the tasks call for the conjunction of several equally important attributes, and lacking one required attribute for the item is the same as lacking all the required attributes. 
\label{ex-dina}
\end{example}

\begin{example}[ACDM] 
In the ACDM, mastering additional required attributes will increase the positive response probability for the items.
Specifically, if we take the identity link function in the ACDM, then for the $j$th item and the $i$th subject with attribute pattern $\bm{\alpha}_i=(\alpha_{i,1},\ldots,\alpha_{i,K})$, we have
 \begin{equation}
    P(R_{i,j}=1\mid\bm{\alpha}_{i})= \delta_{j,0} + \sum_{k=1}^{K}\delta_{j,k}\alpha_{i,k}q_{j,k},
    \label{eq-acdm}
  \end{equation}
which implies that mastering the $k$th attribute increases the probability of success on the item $j$ by $\delta_{j,k}$ if the $k$th latent attribute is required by the item $j$. 
Since there are no interaction terms in \eqref{eq-acdm}, the contribution of each latent attribute is independent from one another. 
If the subject $i$ lacks all the required attributes for the item $j$, the term $\sum_{k=1}^{K}\delta_{j,k}\alpha_{i,k}q_{j,k}$ would be 0, and the intercept $\delta_{j,0}$ is the probability of correctly answering the item $j$ based on pure guessing.
Furthermore, even if the $i$th subject has all the required latent attributes of the item $j$, $\delta_{j,0} + \sum_{k=1}^{K}\delta_{j,k}\alpha_{i,k}q_{j,k}$ may not sum to $1$. 
In that case, $1-\big(\delta_{j,0} + \sum_{k=1}^{K}\delta_{j,k}\alpha_{i,k}q_{j,k}\big)$ would be the probability of making a careless mistake.
The ACDM is more appropriate to use when the items call for independent latent attributes but with different contributions to correct response to the items.

Besides the identity link function, other link functions are also proposed. 
One commonly used link function is the logit link,
 \begin{equation}\label{eq:Logit_ACDM}
    P(R_{i,j}=1\mid\bm{\alpha}_{i})= \sigma\Big(\delta_{j,0} + \sum_{k=1}^{K}\delta_{j,k}\alpha_{i,k}q_{j,k}\Big).
  \end{equation}
where $\sigma(x)=(1+\exp(-x))^{-1}$. 
Equation \eqref{eq:Logit_ACDM} is also equivalent to $\text{logit}\big(P(R_{i,j}=1\mid \bm{\alpha}_{i})\big)=\delta_{j,0} + \sum_{k=1}^{K}\delta_{j,k}\alpha_{i,k}q_{j,k}$, which is the log-odds of a positive response. 
The interpretation would then become that each required latent attribute contributes independently to the log-odds of correcting answering item $j$ by $\delta_{j,k}$ in an additive fashion.
\label{ex-acdm}
\end{example}

\begin{example}[GDINA model] 
Both the DINA and ACDM models are special cases of the more general GDINA model \citep{GDINA}. 
In addition to the intercept and the main effects in the ACDMs, the GDINA model also allows interactions amongst the latent attributes. 
The equation (\ref{eq:GDina}) gives the item response function for the GDINA model with identity link. 
 \begin{equation}\label{eq:GDina}
    P(R_{i,j}=1\mid \bm{\alpha}_{i})= \delta_{j,0} + \sum_{k=1}^{K}\delta_{j,k}\alpha_{i,k}q_{j,k} + \sum_{k=1}^{K-1}\sum_{k^{'}=k+1}^{K}\delta_{jkk'}\alpha_{i,k}\alpha_{i,k'}q_{j,k}q_{j,k'} + ... + \delta_{j12...K}\prod_{k=1}^{K}\alpha_{i,k}q_{j,k}.
  \end{equation}
The parameters in equation \eqref{eq:GDina} can be interpreted as follows:
$\delta_{0}$ is the probability of a correct response when none of the required attributes is present; 
$\delta_{k}$ is the change in the probability of a correct response when only mastering a single attribute $\alpha_{k}$;
$\delta_{kk'}$, a first-order interaction effect, is the change in the probability of a correct response due to the possessing of both $\alpha_{k}$ and $\alpha_{k'}$ in addition to the main effects of mastering the two individual attributes; 
and $\delta_{12...K}$ represents the change in the probability of a correct response due to the mastery of all the required attributes in addition to the main effects and all the lower-order interaction effects. Similarly to the ACDM model, $P(R_{i,j}=1\mid\bm{\alpha}_{i})$ is not required to be 1 even when the subject $i$ possesses all the required attribute for the item $j$. 
In that case, $1-P(R_{i,j}=1\mid\bm{\alpha}_{i})$ is the probability of making a careless mistake. 
Moreover, the intercept $\delta_{j,0}$ and the main effects are typically non-negative, but the interaction effects can take on any values. 
Therefore, the GDINA model is appropriate if the mixed effects of latent attributes on the probability of a correct response is of interest.
\label{ex-gdina}
\end{example}

\subsection{Review of Restricted Boltzmann Machines}
\label{sec-RBM}

RBMs are generative models that can learn probabilistic distributions over a collection of inputs.
RBMs were initially invented under the name Harmonium by \cite{RBM} and gained currency due to their fast learnability in the mid-2000. 
It has found vast applications in dimension reduction \citep{dimReductionRBM}, classification \citep{classificationRBM}, collaborative filtering \citep{collaborativeFilteringRBM} and many other fields.

RBMs can also be viewed as a probabilistic bipartite graphical models, with observed (visible) units in one part of the graph and latent (hidden) units in the other part. 
Typically all the hidden units and the visible units are binary.
In this work, we denote the visible units by $\boldsymbol{R}=\{R_{1}, ...R_{J}\} \in \{0,1\}^{J}$ and hidden units by $\bm{\alpha} = \{\alpha_{1}, ... \alpha_{K}\} \in \{0,1\}^{K}$ respectively. 
One key feature of RBMs is that only interactions between hidden units and visible units are allowed.
There are neither connections among the visible units, nor any connections among the hidden units, as shown in Figure \ref{fig:RBM}. 
 
\begin{figure}[H]
\centerline{\includegraphics[scale=0.55]{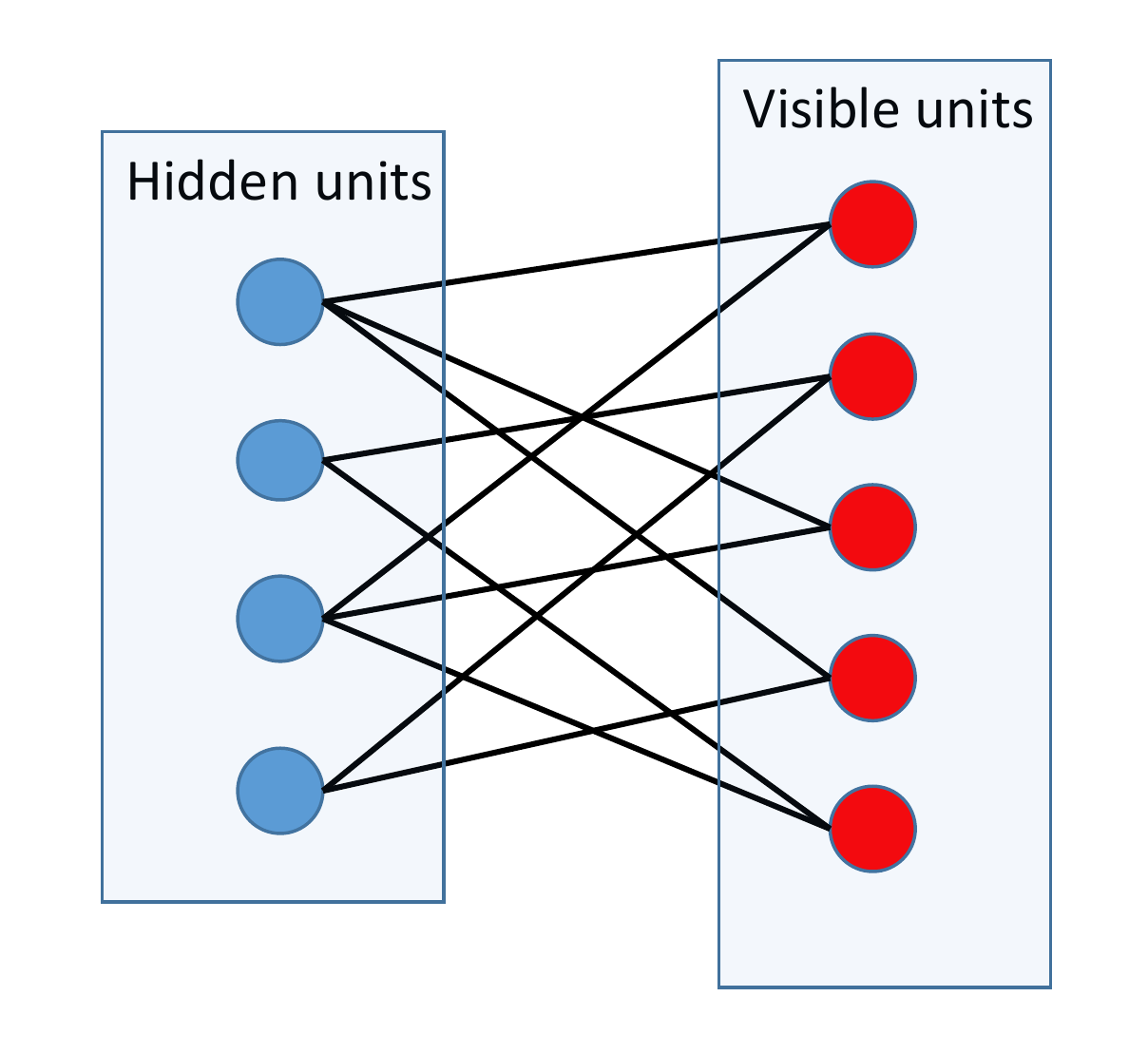}}
\centering
\caption{A graphical illustration of RBM.}
\label{fig:RBM}
\end{figure}

RBMs are characterized by the energy functions with the joint probability distribution specified as
\begin{equation}\label{eq:RBM}
P(\bm{R}, \bm{\alpha}; \bm{\theta}) =\frac{1}{Z(\bm{\theta})}\exp\big\{-E(\bm{R},\bm{\alpha};\bm{\theta})\big\},
\end{equation}
where $E(\bm{R},\bm{\alpha};\bm{\theta})$ is known as the energy function 
and $Z(\bm{\theta})$ is the partition function,
\[
Z(\bm{\theta}) = \sum_{\bm{R}\in \{0,1\}^{J}} \sum_{\bm{\alpha}\in \{0,1\}^K} \exp \big\{-E(\bm{R},\bm{\alpha};\bm{\theta})\big\},
\]
which has been proved to be intractable \citep{long2010restricted}. 
In specific, the energy function is given by
\begin{align*}
E(\bm{R}, \bm{\alpha}; \bm{\theta}) &=
-\bm{b}^{T}\bm{R} - \bm{c}^{T}\bm{\alpha} - \bm{R}^{T}\bm{W}\bm{\alpha}\\
&= {-}\sum_{j=1}^{J}R_{j}b_{j} - \sum_{k=1}^{K} \alpha_{k}c_{k} - \sum_{j=1}^{J}\sum_{k=1}^{K} R_{j}w_{j,k}\alpha_{k}, \numberthis\label{eq:EnergyFunction}
\end{align*}
where $\bm{\theta}=\{\bm{b}, \bm{c}, \bm{W}\}$ are the model parameters, $\bm{b} \in \mathbb{R}^{J}$ are visible biases, $\bm{c} \in \mathbb{R}^{K}$ are hidden biases and $\bm{W} \in \mathbb{R}^{J \times K}$ is the weight matrix describing the interactions between the visible and the hidden units.

Since no ``$\bm{R}$-$\bm{R}$'' or ``$\bm{\alpha}$-$\bm{\alpha}$'' interactions are allowed, the hidden and visible units are conditionally independent given each other, 
and therefore the joint conditional probability mass functions can be factored in to a product. 
This can be easily seen from Equation (\ref{eq:RBM}) and Equation (\ref{eq:EnergyFunction}). 
Specifically, we have
\begin{equation}
    P\big(\bm{R}\mid \bm{\alpha}; \bm{\theta}\big) 
=\prod_{j=1}^J P\big(R_{j} \mid\bm{\alpha};\bm{b},\bm{W}\big),
\numberthis \label{eq:v_given_h}
\end{equation}
\begin{equation}
    P\big(R_{j}=1\mid \bm{\alpha}; \bm{b},\bm{W}\big)
= \sigma\Big(b_{j}+\sum_{k=1}^K w_{j,k}\alpha_{k}\Big),
\numberthis \label{eq:binaryv_given_h}
\end{equation}
and 
\begin{equation}
    P\big(\bm{\alpha}\mid \bm{R}; \bm{\theta}\big) 
=\prod_{k=1}^K P\big(\alpha_{k}\mid\bm{R};\bm{c},\bm{W}\big),
\numberthis \label{eq:h_given_v}
\end{equation}
\begin{equation}
    P\big(\alpha_{k}=1\mid \bm{R}; \bm{c},\bm{W}\big)
= \sigma\Big(c_{k}+\sum_{j=1}^J w_{j,k}R_{j}\Big),
\numberthis \label{eq:binaryh_given_v}
\end{equation}
where $\sigma(x)=1/(1+\exp\{-x\})$ is the logistic sigmoid function. 

RBMs and CDMs are in fact closely related.
The binary observed item responses and the latent attributes in CDMs can be viewed as counterparts to the visible units and the hidden units in RBMs respectively. 
There is a direct connection between the two. 
If we fit an ACDM with the logit link, 
where the conditional probability mass function (\ref{eq:Logit_ACDM}) of the observed responses is modeled as a sigmoid function of the latent attributes, 
then it takes exactly the same form as the conditional probability function (\ref{eq:binaryv_given_h}) of a visible unit given the hidden units in RBMs.
Moreover, in a CDM, $q_{j,k}=0$ indicates that there is no interaction between the item $j$ and the latent attribute $k$, while in the weight matrix of an RBM, $w_{j,k}=0$ also implies no interaction between the $j$th visible unit and the $k$th hidden unit. 
Therefore we would expect that $w_{j,k}=0$ in an RBM whenever $q_{j,k}=0$ in a CDM.

Using the previous example in Figure \ref{fig:RBM} for illustration, on the left of \eqref{WQ} is the weight matrix $\bm{W}$ of an RBM, where $w_{j,k} \neq 0$ indicates the presence of the interaction between the visible unit $R_j$ and the hidden unit $\alpha_k$. 
The corresponding $Q$-matrix in a CDM can be implied as shown on the right. 
As we illustrate previously, the non-zero entries in the $Q$-matrix of an ACDM can be exactly inferred from the non-zero entries in the weight matrix $\bm{W}$ in an RBM.
Interactions among the latent attributes are allowed in the DINA and GDINA models, which violates the assumptions of an RBM. 
However, the $Q$-matrix is still estimable in these models. We give detailed arguments in Section \ref{sec-robust}.

\begin{equation}\label{WQ}
   W=
\begin{bmatrix}
  w_{11}&0&w_{13}&0\\
  0&w_{22}&0&w_{24}\\
  w_{31}&0&w_{33}&0\\
  w_{41}&0&0&w_{44}\\
  0&w_{52}&w_{53}&0
\end{bmatrix}
\quad\implies\quad Q=
\begin{bmatrix}
  1&0&1&0\\
  0&1&0&1\\
  1&0&1&0\\
  1&0&0&1\\
  0&1&1&0
  \end{bmatrix} 
\end{equation}

\subsection{Robust Estimation of Q-matrix}
\label{sec-robust}
In the previous section, we have discussed that RBMs can be used to learn the $Q$-matrix for the ACDM with logit link. 
A natural question to ask is whether we can generalize this result to other CDMs such as the DINA and GDINA models. 
In this section, we will illustrate that under certain conditions, robust estimation of the $Q$-matrix by RBMs is indeed achievable for common CDMs. 
In particular, we will demonstrate that the $Q$-matrix can be estimated correctly under the DINA and GDINA settings. 

We focus on the learning of a particular row of the $Q$-matrix. It is in fact a variable selection problem of the required latent variables for that particular item of interest.
Conditional on $\bm{\alpha}$, we have discussed that  RBMs are equivalent to the ACDM with the logit link, while the latter exactly corresponds to the logistic regression with canonical link and additive main effects linear predictor. 
Therefore in essence, RBMs can also be treated as main effect models. 
Starting with the simplest case, we shall first study the model selection consistency with linear additive models when the true models are the DINA or the GDINA model. 
Since it is still an open and challenging problem to establish consistent variable selection under complex latent variable models, here we start with the ideal case by assuming $\{\alpha_{1},...,\alpha_{K}\}$ are independent, that is, all the latent variables are independent. Although this is a strong assumption and is rarely fully satisfied in real world scenarios, it can be relaxed in practice which is discussed in Remark \ref{rmk-dependency}.

Before giving formal statements, we first introduce some notations.  
Without loss of generality, we focus on the analysis of the response to one single item.
For a subject with $\bm{\alpha}=\{\alpha_{1},...,\alpha_{K}\}$, the response to the considered item is denoted by $R$, where for clarity, we omit the item index in the notation.
Let $K^{*}$ to be the number of required attributes for the item. Without loss of generality, we let the first $K^{*}$ attributes be the required attributes for this item, that is, the corresponding row in the $Q$-matrix is $\textbf{q}=(1,...,1,0,...,0)$ with the first $K^{*}$ entries being 1 and all the remaining $K-K^{*}$ entries being 0. 
For the response $R$ generated from the DINA or the GDINA model, we denote $\mathop{\mathbb{E}}^*[R\mid\bm{\alpha}]$ as  the regression mean function for the mis-specified linear regression model of $R$ on $\alpha_{1},...,\alpha_{K}$. We show in the following propositions that the mis-specified mean function  $\mathop{\mathbb{E}}^*[R\mid\bm{\alpha}]$ can identify the required attributes from the non-required ones.
\begin{proposition}[DINA model]
\label{thm-dina}
Assume $\{\alpha_{1}, \alpha_{2}, ..., \alpha_{K}\}$ are independent with $\alpha_{k} \sim$ Beroulli$(p_{k})$ where $p_{k} \in (0,1)$, $k=1,2,...,K$. If $R$ is generated from the DINA model,
then the mis-specified linear additive model of $R$ regressed on $(\alpha_{1}, \alpha_{2}, ..., \alpha_{K})$ has the mean function in the form of $\mathop{\mathbb{E}}^*[R\mid\bm{\alpha}]=\beta_0+\beta_{1}\alpha_{1}+\beta_{2}\alpha_{2}+...+\beta_{K}\alpha_{K}$ with  $\beta_{l} \neq 0$ for $l=1,2,...,K^{*}$ and $\beta_{k}=0$ for $k= K^{*}+1,...,K$.
\end{proposition}
Proposition \ref{thm-dina} states that under the independence condition and if the data is generated from the DINA model, 
the significant variables included in the true model can be selected correctly using a mis-specified linear model with additive main effects only.

\begin{proposition}[GDINA model]
Assume $\{\alpha_{1}, \alpha_{2}, ..., \alpha_{K}\}$ are independent with $\alpha_{k} \sim$ Bernoulli$(p_{k})$ where $p_{k} \in (0,1)$, $k=1,2,...K$. 
If $R$ is generated from the GDINA model
satisfying the monotonicity assumption (i.e. acquiring an additional required skill $\alpha_{k}$, $k=1,2,..,K^{*}$, will always increase the probability of a correct response), then the mis-specified linear additive model has the corresponding mean function in the form of $\mathop{\mathbb{E}}^*[R\mid\bm{\alpha}]=\beta_0+\beta_{1}\alpha_{1}+\beta_{2}\alpha_{2}+...+\beta_{K}\alpha_{K}$ with  $\beta_{l} \neq 0$ for $l=1,2,...,K^{*}$ and $\beta_{k}=0$ for $k= K^{*}+1,...,K$.
\label{thm-gdina}
\end{proposition}

Similar to Proposition \ref{thm-dina}, Proposition \ref{thm-gdina} states that under suitable conditions, the significant variables included in the true GDINA model can be selected correctly using a mis-specified linear model with additive main effects only.
The detailed proofs for all the propositions can be found in Section \ref{sec-appendix-proofs} of the Supplementary Materials.

Propositions \ref{thm-dina} and  \ref{thm-gdina} demonstrate that the model selection consistency can be achieved using a mis-specified linear main effect model.
As we illustrated previously, the conditional probability of a visible unit on the hidden units in RBMs can be regarded as a main effect logistic regression model.
Therefore we next give some intuition on why the main effect logistic regression model will give a similar variable selection result to the linear models. 
Consider a main effect logistic regression model with the canonical link function, that is, $\mbox{logit}\big(P(R\mid \bm{\alpha})\big)=\beta_0+\beta_{1}\alpha_{1}+...+\beta_{K}\alpha_{K}$.
Let ${\cal R} = (R_i, i=1,...,N)$ denote the response vector for all the $N$ subjects,
and let $\bm{\mu}=\big(\mu_{i}:=P(R_{i}\mid \bm{\alpha}_i\big), i=1,...,N)$ denote the response probabilities for the subjects.
We use $\bm A = \big(\bm{\alpha}_{i}\big)_{i=1}^N \in \{0,1\}^{N\times K}$ to denote the latent attribute matrix for the $N$ subjects and $\bm A^*$ to denote the $N\times (K+1)$ matrix $[\mathbf{1}; \bm A ]$ with the first column being an all-one vector.
In linear models, we usually use the least square estimation to estimate the coefficients,
while in logistic regression, the iteratively re-weighted least square (IRLS) method is used.
Next we will give some intuition on why these two estimation methods will produce similar variable selection results. 

Conditional on $\bm{\alpha}_i$'s, in the $(t+1)$th step of IRLS, the updating rule for parameter $\bm{\theta}:=(\beta_0,\beta_{1},...,\beta_{K})$ is 
\[
\bm{\theta}^{(t+1)}= \big({\bm{A}^*}^{T}\bm{W}^{(t)}{\bm{A}^*}\big)^{-1}{\bm{A}^*}^{T}\bm{W}^{(t)}\bm{Z}^{(t)},
\]
where $\bm{Z}^{(t)}={\bm{A}^*}^{T}\bm{\theta}^{(t)}+(\bm{W}^{(t)})^{-1}({\cal R}-\bm{\mu}^{(t)})$ is the $t$th step working response and $\bm{W}^{(t)}=\text{diag}\big(\mu_{1}^{(t)}(1-\mu_{1}^{(t)}),...,\mu_{N}^{(t)}(1-\mu_{N}^{(t)})\big)$ is a diagonal weight matrix with diagonal elements being the variance estimates for each $R_i$. 
Since there is no closed form of IRLS estimator and there is randomness in the convergence process, it is very challenging to study the theoretical properties of the $\bm{\theta}$ estimated by IRLS. 
So we only consider a one-step update of IRLS starting from the ideal case of true parameter $\bm{\theta}_{\text{true}}$ for illustration. 
It is reasonable to study this ideal case because IRLS will converge close to the $\bm{\theta}_{\text{true}}$ given the correct model specification and a large sample size. If we start with the true parameters, that is, we let $\bm{\theta}^{(0)}=\bm{\theta}_{\text{true}}$, then,
\[
\bm{\theta}^{(1)}= \big({\bm{A}^*}^{T}\bm{W}_{\text{true}}{\bm{A}^*}\big)^{-1}{\bm{A}^*}^{T}\bm{W}_{\text{true}}\bm{Z}_{\text{true}}, 
\]
where the working response, $\bm{Z}_{\text{true}}={\bm{A}^*}^{T}\bm{\theta}_{\text{true}}+\bm{W}_{\text{true}}^{-1}({\cal R}-\bm{\mu}_{\text{true}})$ is just a linear transformation of observed response ${\cal R}$. 
Note that this update takes the same form as the weighted least square estimation of regressing $\bm{Z}_{true}$ on $\bm{A}^*$. 
Hence, the variable selection result in the linear model would be similar to that of the logistic regression. 
Combining Proposition \ref{thm-dina} and Proposition \ref{thm-gdina}, we have justified that the learning of the $Q$-matrix by RBMs is achievable across the DINA, ACDM and GDINA models with both identity and logit links.

 



\begin{remark}
In practice, it is not uncommon that some of the $2^{K}$ latent attribute patterns do not exist in the collected observations, especially when $K$ is large. 
How negatively will this impact on the model selection consistency? 
In the DINA model, we see from the proof of Proposition \ref{thm-dina} (see Section \ref{sec-appendix-proofs} of the Supplementary Materials) that to ensure the variable selection consistency for each required attribute $\alpha_k$, $k=1,...,K^{*}$, 
we need to observe data from subjects with $\big\{\bm{\alpha}\mid \alpha_{k}=0, \alpha_{i}=1, i = 1,\dots, k-1, k+1,\dots, K^*\big\}$
and $\big\{\bm{\alpha}\mid \alpha_{i}=1, i =1,\dots,K^*\big\}$. 
In the GDINA model, from the proof of Proposition \ref{thm-gdina}, we can see that  to ensure the variable selection consistency for each $\alpha_k$, $k=1,...,K^{*}$, we need to observe data from subjects with $\big\{\bm{\alpha}\mid \alpha_{k}=0\big\}$ and $\big\{\bm{\alpha} \mid \alpha_{k}=1\big\}$. 
Therefore, even though some of the latent patterns may not exist in our observed data, the selection consistency is still achievable as long as the required attribute patterns are present.
\end{remark}
\begin{remark}\label{rmk-dependency}
The independent assumption on the latent attributes $\{\alpha_1,...,\alpha_K\}$  can be relaxed to some extent in practice. 
To see this, consider the setting when $\{\alpha_1,...,\alpha_{K}\}$ are possibly dependent but the response $R$ only directly depends on the first $K^*$ attributes $\{\alpha_1,...,\alpha_{K^*}\}$.
Given $\alpha_1,...,\alpha_{K^*},$ the response $R$ is conditionally independent of $\alpha_k$ for all $k=K^*+1,...,K.$
When only $\alpha_1,...,\alpha_{K^*}$ are present in the linear regression model of $R$ regression on $\alpha$'s, consider adding in one additional $\alpha_k$, for any $k=K^*+1,...,K,$ into the regression model, then its coefficient can be expressed as

\begin{align*}
\beta_k
=\frac{Cov\Big(R-\mathbb{E}^*[R\mid \alpha_1,...,\alpha_{K^*}],\quad \alpha_k-\mathbb{E}^*[\alpha_k\mid\alpha_1,...,\alpha_{K^*}]\Big)}{Var\Big(R - \mathbb{E}^*[R\mid \alpha_1,...,\alpha_{K^*}]\Big)},\numberthis\label{eq: cov-residual}
\end{align*}
where we denote $\mathbb{E}^*[A\mid B]$ as the regression mean function of $A$ on $B$. 
Since $R$ and $\alpha_k$ are conditionally independent given $\alpha_1,...,\alpha_{K^*},$ the numerator of \eqref{eq: cov-residual} is expected to be small. 
In real implementations, the shrinkage imposed by the $L_1$ penalty in our proposed method should be able to recover most of these 0's.
This is indeed supported by our simulation results in Section \ref{sec-simulation-studies}, where we consider moderate to high correlation regimes amongst the latent attributes and our proposed method still achieves satisfactory estimation accuracy of the underlying $Q$-matrix. 
Note also that in the special case when $K^*=1,$ the covariance term in \eqref{eq: cov-residual} can be shown exactly equal to zero, in which case $\beta_k$ can be removed easily in the variable selection process. 
For a more detailed discussion for the $K^*=1$ case, please refer to Section \ref{sec-appendix-proofs} of the Supplementary Materials.


\end{remark}

\begin{remark}


The rigorous consistency theory of using RBMs to learn the $Q$-matrix under a general CDM setting can be difficult to establish.
In the literature, even when the true models are binary RBMs, 
consistency for training RBMs is  an open and challenging problem.
 Due to the intractable partition function in the binary RBM, an approximate likelihood maximizing approach has to be employed, such as the popularly used Contrastive Divergence (CD) algorithm that will be further introduced in Section \ref{sec-method&algo}.
  Even though there are many works in literature studying the asymptotic properties of the CD algorithm \citep{mackay2001failures,yuille2004convergence,carreira2005contrastive,bengio2009justifying, sutskever2010convergence,jiang2018convergence},
 whether  and  why  the  CD  algorithm  provides  an  asymptotically  consistent  estimate for  binary  RBMs  are  still  open  questions.
 Therefore, establishing a consistency theorem using a mis-specified RBM model for the DINA or the GDINA model as in this work is even more challenging, which is left for future exploration.
  Nevertheless, the CD algorithm in practice has showed empirical success in training RBMs, and our simulation results in Section \ref{sec-simulation-studies} also demonstrate its effectiveness in training RBMs to learn the $Q$-matrix in CDMs.

\end{remark}

\section{Proposed Estimation Method}
\label{sec-method&algo}

In this section, we will introduce our proposed method in detail.
As we have illustrated in Section \ref{sec-CDM&RBM},
non-zero entries in the $Q$-matrix can be inferred from the corresponding non-zero entries in the weight matrix of RBMs.
Therefore, we are interested in a sparse solution of the weight matrix $\bm{W}$.
It is well known that $L_1$ penalty has the property of producing sparse solutions \citep{Rosasco2009}.
Hence, we propose the following $L_1$ penalized likelihood as our objective function,
\begin{equation} \label{ eq:objective}
    \min_{\bm{\theta}} -\log\big\{ P(\bm{R};\bm{\theta})\big\} + \lambda\sum_{j=1}^{J}\sum_{k=1}^{K}|w_{j,k}|.
\end{equation}
where $\log\{ P(\bm{R};\bm{\theta})\}$ is the marginal log-likelihood of the observed responses $\bm{R}$,  $\bm{\theta}=\{\bm{b},\bm{c}, \bm{W}\}$ are the model parameters, 
and $\lambda$ is a non-negative tuning parameter for the $L_1$ penalty.


Gradient descent algorithm is a standard numerical method to solve problem (\ref{ eq:objective}).
The likelihood part, following the derivation by \cite{RBMDerivations}, can be shown that its gradient with respect to the parameters has the following decomposition:
\begin{align}\label{eq:loglikelihoodgradient}
 \frac{\partial}{\partial \bm{\theta}} \log\big(P(\bm{R};\bm{\theta})\big)
 &=-\sum_{\bm{\alpha}\in\{0,1\}^K}P\big(\bm{\alpha}|\bm{R}; \bm{\theta}\big) \frac{\partial}{\partial \bm{\theta}} E\big(\bm{R}, \bm{\alpha} ; \bm{\theta}\big)+
 \sum_{\substack{\bm{r} \in \{0,1\}^J \\\bm{\alpha}\in\{0,1\}^K}}
 P\big(\bm{r}, \bm{\alpha};\bm{\theta}\big)\frac{\partial}{\partial \bm{\theta}}E\big(\bm{r}, \bm{\alpha} ; \bm{\theta}\big)\\
 &=\mathbb{E}_{P(\bm{\alpha|R};\bm{\theta})}\Big[-\frac{\partial}{\partial \bm{\theta}} E\big(\bm{R}, \bm{\alpha}; \bm{\theta}\big)\Big] - \mathbb{E}_{P(\bm{r,\alpha};\bm{\theta})}\Big[- \frac{\partial}{\partial \bm{\theta}}E\big(\bm{r}, \bm{\alpha}; \bm{\theta}\big) \Big].\label{eq:decp}
\end{align}
In deep learning literature, this is a well-known decomposition into the positive phase and the negative phase of learning, corresponding to the two expectations in \eqref{eq:decp} respectively.  As the two expectations do not have closed forms and are not directly tractable, researchers propose to  approximate the gradient by estimating these expectations through Monte Carlo sampling. In particular, the positive phase corresponds to sampling the hidden units given the visible units,  while the negative phase corresponds to obtaining the joint hidden and visible samples from the current model. 

The bipartite graph structure of RBMs gives the special property of its conditional  distributions $P(\bm{\alpha}\mid\bm{R})$ and $P(\bm{R}\mid\bm{\alpha})$ being factorial and simple to compute and sample from, as shown in Section \ref{sec-RBM}.
Therefore, sampling for the positive phase is straightforward while obtaining samples from the model for negative phase is not since it requires the joint hidden and visible samples.
A widely used algorithm to learn RBMs is known as the Contrastive Divergence (CD) algorithm, where the negative phase is approximated by drawing samples from a short alternating Gibbs Markov chain between visible units and hidden units starting from the observed training examples \citep{hinton2002training}.
In this work, we use a CD-1 algorithm where Gibbs chains are run for 1 step to approximate the gradient of the log-likelihood part.
Specifically, given the original data $\bm{R}^{(0)}$, we first sample $\bm{\alpha}^{(0)}$ according to Equation (\ref{eq:h_given_v}) and Equation (\ref{eq:binaryh_given_v}) to approximate the positive phase.
Then given $\bm{\alpha}^{(0)}$, we sample $\bm{R}^{(1)}$ based on Equation (\ref{eq:v_given_h}) and Equation (\ref{eq:binaryv_given_h}), and we use $(\bm{R}^{(1)},\bm{\alpha}^{(0)})$ to approximate the negative phase.

At $(t+1)$th iteration, based on the sampled data, the parameters' updates take the same form as gradient descent if we do not consider $L_1$ penalty,
\begin{align}
    \label{eq-updates-W}
    {w'}_{j,k}^{(t+1)} & \leftarrow w_{j,k}^{(t)} + \gamma^{(t)}\Big\{
    \sum_{i=1}^N R_{ij}^{(0)}P\big(\alpha_{ik}=1\mid\bm{R}_i^{(0)};\bm{\theta}^{(t)}\big)
    - \sum_{i=1}^N R_{ij}^{(1)} P\big(\alpha_{ik}=1\mid\bm{R}_i^{(1)};\bm{\theta}^{(t)}\big)\Big\},\\
    \label{eq-update-b}
    b_{j}^{(t+1)} & \leftarrow b_{j}^{(t)} + \gamma^{(t)}\Big\{\sum_{i=1}^N R_{i,j}^{(0)}-\sum_{i=1}^{N}R_{i,j}^{(1)}\Big\}/N,\\   
    \label{eq-update-c}
    c_{k}^{(t+1)} & \leftarrow c_{k}^{(t)} + \gamma^{(t)}\Big\{\sum_{i=1}^N P\big(\alpha_{ik}=1\mid\bm{R}_i^{(0)};\bm{\theta}^{(t)}\big)
    -\sum_{i=1}^N P\big(\alpha_{ik}=1\mid\bm{R}_i^{(1)};\bm{\theta}^{(t)}\big)\Big\}/N,
\end{align}
where $\bm{R}_{i}^{(0)} = (R_{i1}^{(0)},R_{i2}^{(0)},\dots,R_{iJ}^{(0)})$, $\bm{R}_i^{(1)} = (R_{i1}^{(1)}, R_{i2}^{(1)},\dots,R_{iJ}^{(1)})$, and $\gamma^{(t)}$ is the learning rate for the $t$th iteration.
Here we denote the updated weight matrix by $\bm{W}'= \big(W'_{j,k}\big)_{J\times K}$, since we also need to consider the gradient of the $L_1$ penalty term later, and thus Equation (\ref{eq-updates-W}) is an intermediate update for the weight matrix.
Detailed derivations can be found in the notes written by \cite{RBMDerivations}. 
In this work, we use a linearly decreasing learning rate scheme, which is guaranteed to converge as shown in \cite{learningRate}.

For the $L_1$ penalty term, we adopt the implementation developed by \cite{L1}, which can achieve more stable sparsity structures. 
As pointed out by \cite{L1}, the traditional implementation of $L_1$ penalty in gradient descent algorithm does not always lead to sparse models because the approximate gradient used at each update is very noisy, which deviates the updates away from zero.

The main idea of the implementation is to keep track of the total penalty and the penalty that has been applied to each parameter,
and then the $L_1$ penalty is applied based on the difference between these cumulative values. 
By doing so, it is argued that the effect of noisy gradient is smoothed away. 
To be more specific, at iteration $t$, let $u^{(t)}:=\lambda\sum_{l=1}^{t}\gamma^{(l)}$ be the absolute value of the total $L_1$ penalty that each parameter could have received up to the point, where $\gamma^{(l)}$ is the learning rate at step $l$. Let $c_{j,k}^{(t-1)}:=\sum_{l=1}^{t-1}(w_{j,k}^{(l+1)}-{w'}_{j,k}^{(l+1)})$ be the total $L_1$ penalty that $w_{j,k}$ has actually received up to step $t$,
where ${w'}_{j,k}^{(l)}$ is the intermediate update at step $l$ calculated by Equation (\ref{eq-updates-W}).
Then at iteration $(t+1)$, we update $w_{j,k}^{(t+1)}$ by
\[    w_{j,k}^{(t+1)} \leftarrow \max\Big\{0, {w'}_{j,k}^{(t+1)}-(u^{(t)}+c_{j,k}^{(t-1)})\Big\} \quad
  \text{if} \quad  {w'}_{j,k}^{(t+1)} > 0,\]
  \[
  w_{j,k}^{(t+1)} \leftarrow \min\Big\{0,  {w'}_{j,k}^{(t+1)}+(u^{(t)}-c_{j,k}^{(t-1)})\Big\} \quad \text{if} \quad  {w'}_{j,k}^{(t+1)} \leq 0.
\]


Since the updates in Equation (\ref{eq-updates-W}), (\ref{eq-update-b}) and (\ref{eq-update-c}) require summations over all the data samples,
it would be computationally expensive when the sample size is large.
To reduce computational burden, we implement a batch version of the CD-1 algorithm in practice, where we only use a small batch of the whole data set in each iteration.
Specifically, we randomly partition the whole data set into $B$ batches, and iterating through all the batches is known as one epoch in machine learning literature.
Here we use $\bm{R} = \big\{\bm{R}_{(1)},\bm{R}_{(2)},\dots,\bm{R}_{(B)}\big\}$ to denote the partitions, $N_B$ to denote the batch size, and $N_{\text{epoch}}$ to denote the number of epoches.
The resulting algorithm is summarized in Algorithm \ref{algo}.

\begin{algorithm}[h!]
\SetAlgoLined

\KwIn{Data $\bm{R}=\Big\{\bm{R}_{(1)},\bm{R}_{(2)},\dots,\bm{R}_{(B)}\Big\}$, $\lambda$, $\gamma_0$, and $N_{\text{epoch}}$.}

\KwOut{Estimates $\hat{\bm{W}}$, $\hat{\bm{b}}$, $\hat{\bm{c}}$.}

Initialize $w_{j,k}^{(0)}, b_{j}^{(0)}, c_{k}^{(0)}$,
$u^{(0)}=0$, $c_{j,k}^{(0)}=c_{j,k}^{(1)}=0$\;

\For{$e = 0,\dots, N_{\text{epoch}}-1$}{
    \For{$b = 0, \dots, B-1$}{
        $t = e\times B + b $ (the number of iterations)\;
        $\gamma^{(t)}=\frac{\gamma_{0}}{t+1}$\;
        $\bm{R}^{(0)} \leftarrow \bm{R}_{(b+1)}$\;
        Sample $\bm{\alpha}^{(0)} \sim P\big(\bm{\alpha}\mid\bm{R}^{(0)};\bm{c}^{(t)},\bm{W}^{(t)}\big)$\;
        Sample $\bm{R}^{(1)} \sim P\big(\bm{R}\mid\bm{\alpha}^{(0)};\bm{b}^{(t)},\bm{W}^{(t)}\big)$\;
        $u^{(t)} \leftarrow u^{(t-1)} + \lambda \gamma^{(t)}$\;
        \For{$j=1,\dots,J, k=1,\dots,K$}{
            ${w'}_{j,k}^{(t+1)} \leftarrow w_{j,k}^{(t)} + \gamma^{(t)}\Big\{
            \sum_{i=1}^{N_B} R_{ij}^{(0)}P\big(\alpha_{ik}=1\mid\bm{R}_i^{(0)}\big)
            - \sum_{i=1}^{N_B} R_{ij}^{(1)} P\big(\alpha_{ik}=1\mid\bm{R}_i^{(1)}\big)\Big\}$;
            \ \\
            \lIf{$t\geq 2$}{ $c_{j,k}^{(t-1)} \leftarrow c_{j,k}^{(t-2)} + w_{jk}^{(t)}-{w'}_{j,k}^{(t)}$} 
            \ \\
            \uIf{${w'}_{j,k}^{(t+1)} > 0$}{
                $w_{j,k}^{(t+1)} \leftarrow \max\Big\{0, {w'}_{j,k}^{(t+1)}-(u^{(t)}+c_{j,k}^{(t-1)})\Big\}$\;
            }
            \uElse{
                $w_{j,k}^{(t+1)} \leftarrow \min\Big\{0,  {w'}_{j,k}^{(t+1)}+(u^{(t)}-c_{j,k}^{(t-1)})\Big\} $ \;
            }
        }
    }    
            \ \\
            \For{$j=1,...,J$}{
                $b_{j}^{(t+1)}  \leftarrow b_{j}^{(t)} + \gamma^{(t)}\Big\{\sum_{i=1}^{N_{B}}R_{i,j}^{(0)}-\sum_{i=1}^{N_{B}}R_{i,j}^{(1)}\Big\}/N_{B}$\;
            }
            \ \\
            \For{$k=1,...,K$}{
                $c_{k}^{(t+1)}  \leftarrow c_{k}^{(t)} + \gamma^{(t)}\Big\{\sum_{i=1}^{N_B}P\big(\alpha_{ik}=1\mid\bm{R}_i^{(0)}\big)               -\sum_{i=1}^{N_B}P\big(\alpha_{ik}=1\mid\bm{R}_i^{(1)}\big)\Big\}/N_B$\;
            }
}
\caption{CD-1 algorithm with $L_1$ penalty}
\label{algo}
\end{algorithm}

In our proposed algorithm, there are two tuning parameters: $\lambda$ for the $L_1$ penalty and $\gamma_0$ for the learning rate.
To get good estimates of our model, we need to select a suitable combination of hyper-parameters $\lambda$ and $\gamma_{0}$.
A popularly used tuning procedure is cross validation (CV).
However, as RBMs are unsupervised learning models, we cannot rely on the so-called ``test error" of the labels. Instead, since visible units are re-sampled at each iteration in the CD algorithm,
we may use the reconstruction error of the visible units to assess the goodness of fit. 
Nevertheless, the visible reconstruction error will always increase as the penalty coefficient $\lambda$ increases, because larger penalty would introduce more bias. Therefore, the traditional CV procedures would not work here. 
To solve this problem, given values of $\lambda$ and $\gamma_0$, instead of directly using the $\hat{\bm{W}}_{\lambda,\gamma_0}$ obtained from a penalized RBM to compute the reconstruction error, 
we propose to debias the non-zero entries in $\hat{\bm{W}}_{\lambda,\gamma_0}$ by training an RBM with no penalty but fixing the zero positions the same as those in $\hat{\bm{W}}_{\lambda,\gamma_0}$.
The proposed CV producedure is summarized below.
\begin{enumerate}
    \item Split the data into $M$ partitions. Each time we use one partition as the validation set and the remaining as the training set.
    
    \item Apply the penalized CD Algorithm \ref{algo} to train the RBM on the training set with pre-specified $\lambda$ and $\gamma_{0}$, and obtain the estimates $\hat{\bm{W}}_{\lambda,\gamma_{0}}$ and $\hat{\bm{Q}}_{\lambda,\gamma_{0}}$.
    
    \item Use the training set again to debias the non-zero entries of $\hat{\bm{W}}_{\lambda,\gamma_{0}}$. 
    Specifically, we use $\hat{\bm{W}}_{\lambda,\gamma_{0}}$ as the initial value and set $\lambda=0$ in Algorithm \ref{algo} to train an unpenalized RBM, and only update the non-zero entries of $\hat{\bm{W}}_{\lambda,\gamma_{0}}$ while keeping the zero entries unchanged. 
    Hidden bias $\bm{c}$ and visible bias $\bm{b}$ are updated at each step as usual. 
    This step give us the de-biased weight matrix $\check{\bm{W}}_{\lambda,\gamma_{0}}$.
    
    \item Compute the reconstruction error on the validation set. 
    In specific, at each iteration of the CD algorithm, we fix $\bm{W} = \check{\bm{W}}_{\lambda,\gamma_{0}}$, and only update the hidden and visible biases. 
    The reconstruction error is computed as the mean batch squared error between the latest sampled visible batches $\{\bm{R}^{(1)}_{1}, ..., \bm{R}^{(1)}_{m}\}$ and the observational batches $\{\bm{R}^{(0)}_{1}, ..., \bm{R}^{(0)}_{m}\}$ in the validation set.
    
    \item For each combination of $\lambda$ and $\gamma_{0}$ in the candidate set, we repeat Step 2-4 across all $M$ validation sets.
    The $\hat{\bm{Q}}_{\lambda^*,\gamma_{0}^*}$ corresponding to the smallest mean batch squared error (see Section \ref{sec-simulation-studies} for definition) is taken as the final estimate of the $Q$-matrix.
    
\end{enumerate}

Another major difference from the traditional CV procedure is we select the $Q$-matrix corresponding to the smallest validation error instead of taking average of the validation errors and then training a new RBM with the best tuning parameters according to the smallest mean error.
There are two advantages.
On one hand, the traditional way of averaging errors, though more stable, is very time-consuming in this problem.
On the other hand, the gradient descent steps in the CD algorithm may only produce locally optimal results.
To avoid being stuck in sub-optima, we run the CD algorithm $M$ times with different initializations and different training and validation sets for each combination of $\lambda$ and $\gamma_{0}$, 
and select the estimated $Q$-matrix corresponding to the smallest validation error.
By doing so, the $Q$-matrix is expected to be more accurately estimated.

\begin{remark}

The computational cost of our proposed method only grows linearly in $K$ and this enables estimation of very large $Q$-matrices.
As far as we know,  the current methods in the literature have computational cost greater than $O(K),$ with the majority growing exponentially with $K$. 
For example, in \cite{XuEstimateQmatrix}, they proposed to learn the $Q$-matrix by estimating the coefficients in the LCDM plus a penalty term with the EM algorithm.
In the E-step of the EM algorithm, $2^K$ posterior probabilities for each of the attribute patterns need to be updated.
However, we also point out that there may be alternative approaches that are also computationally feasible.
Thanks to one of the reviewers, who suggests it may also be feasible to use the traditional ACDM to learn large $Q$-matrices. 
Note that all of our arguments in Sections \ref{sec-robust} also apply to the ACDM model.
With a high-order model that parameterizes the distribution of the binary vector of attributes, such as a Probit model, the number of parameters that need to be learned can be reduced from $2^K$ to $O(K^2)$. Together with a stochastic gradient descent algorithm, this can also be a computationally feasible approach.

\end{remark}

\section{Simulation Studies}
\label{sec-simulation-studies}
We conduct simulation studies on three popular CDMs, the DINA, ACDM and GDINA models, to study the performance of our proposed method in learning the $Q$-matrix under different CDM settings. In particular, we examine the scalability to the size of the $Q$-matrix and the estimation accuracy of the proposed algorithm.

We first introduce the metrics used to evaluate the performance of the proposed estimation method. 
To measure the convergence of the algorithm, we investigate the change in the mean batch errors against time.  
The mean batch error is the reconstruction error between the latest sampled visible batches $\big\{\bm{R}^{(1)}_{(1)}, ..., \bm{R}^{(1)}_{(B)}\big\}$ and the original observed batches $\big\{\bm{R}_{(1)}^{(0)}, ...,\bm{R}_{(B)}^{(0)}\big\}$, where $\big\{\bm{R}_{(1)}^{(0)}, ...,\bm{R}_{(B)}^{(0)}\big\}$ partitions the whole observed data set into $B$ batches.
Given the batch-size $N_B$, the mean batch error is defined as 
\[
\frac{1}{B N_B}\sum_{b=1}^{B}\sum_{i=1}^{N_B}\sum_{j=1}^{J} \Big( R_{(b),i,j}^{(1)}-R_{(b),i,j}^{(0)} \Big)^2 .
\]
To evaluate the estimation accuracy, we report entry-wise overall percentage error (OE), out of true positives percentage error (OTP) and out of true negatives percentage error (OTN). 
Specifically, 
$$\text{OE}:= \frac{1}{JK}\sum_{j=1}^{J}\sum_{k=1}^{K}\mathbbm{1}\big\{\hat{q}_{j,k} \neq q_{j,k}\big\},$$ 
which is the percentage of wrongly estimated entries out of the total number of entries in the $Q$-matrix.
$$\text{OTP}:=\frac{\sum_{j=1}^{J}\sum_{k=1}^{K}\mathbbm{1}\big\{\hat{q}_{j,k}=0, q_{j,k}=1\big\}}{\sum_{j=1}^{J}\sum_{k=1}^{K}\mathbbm{1}\big\{q_{j,k}=1\big\}},$$ 
which is defined as the percentage of wrongly estimated entries out of all true positive entries (i.e. entries 1) in the $Q$-matrix. 
$$\text{OTN}:=\frac{\sum_{j=1}^{J}\sum_{k=1}^{K}\mathbbm{1}\big\{\hat{q}_{j,k}=1, q_{j,k}=0\big\}}{\sum_{j=1}^{J}\sum_{k=1}^{K}\mathbbm{1}\big\{q_{j,k}=0\big\}},$$ 
which is defined as the percentage of wrongly estimated entries out of all true negatives (i.e. entries 0) in the $Q$-matrix. 
A challenge in computing these errors arises because the estimated $Q$-matrix can only be identified up to column permutations. 
To resolve this problem, we apply the Hungarian algorithm to match the columns of the estimated $\hat{Q}$ to the true $Q$-matrix by jointly minimizing the total column-wise matching errors. Details of the Hungarian algorithm can be found in \cite{Kuhn2010TheHM}.

We consider different number of latent attributes $K=5,10,15,20,25$.
To ensure the $Q$-matrix is identifiable so that it can be learned from the observational data, we specify it as follows: 
\begin{align} \label{eq:Qmatrix}
    Q &= \begin{bmatrix}
           I_{K} \\
           Q_{1} \\
           Q_{2}
           \end{bmatrix},
\end{align}
where $I_{K}$ is a $K$ dimensional identity matrix; 
$Q_{1} \in \{0,1\}^{K \times K}$ with value $1$ in the $(i,i)$th entries for $i=1,...,K$ and the $(i,i+1)$th entries for $i=1,...,K-1$, and values $0$ for all the other entries; 
$Q_{2} \in \{0,1\}^{K \times K}$ with value 1 in entries $(i,i)$ for $i=1,...,K$, $(i, i-1)$ for $i=2,...,K$ and $(i,i+1)$ for $i=1,...,K-1$, and value 0 for all the remaining entries. 
The above construction sets the number of items to be $J=3K$.
This $Q$-matrix satisfies the identifiability conditions in \cite{identifiability} and therefore is identifiable under the DINA setting in Simulation Study \ref{simu-dina}. 
Moreover, this construction also ensures the (generic) identifiability of the ACDM and GDINA models considered in Simulation Studies \ref{simu-acdm} and \ref{simu-gdina} \citep[see][]{xu2017, gu2018sufficient, gu2018partial}. A random design of the $Q$-matrix, in which its identifiability is not be guaranteed, is also considered in Section \ref{sec-appendix-randomQ} of the Supplementary Materials.

In each simulation study, we consider two different sample sizes $N= 2000 \text{ or } 10000$. Both independent and dependent settings of latent attributes are explored. 
Denote the latent attribute matrix by $\bm{A} = \big(\bm{\alpha}_i\big)_{i=1}^N \in \{0,1\}^{N \times K}$, which depicts the latent attribute patterns of the $N$ examinees.
We use two steps to simulate the latent patterns \citep{chen_liu_Xu_ying2015}. 
First, a Gaussian latent vector is generated for each subject $\bm{z}_i=(z_{i1},...,z_{iK}) \stackrel{i.i.d.}{\sim} \mathcal{N}(0,\Sigma)$ for $i=1,...,N$, where $\Sigma=(1-\rho)\bm{1}_K+\rho\bm{1}_K\bm{1}_K^{\top}$, $\bm{1}_K = (1,\dots,1)_K^{\top}$, and $\rho$ is the correlation between any two different latent attributes. In practice, since some attributes may be harder to master than others, different thresholds are applied in the sampling of attribute profiles. In particular, for a given $K$, we specify the thresholds ranging from $-0.5$ to $0.5$, with a step size of $1/(K-1)$, for each of attribute $1, 2, ... K$ respectively. Then $\alpha_{ik}=1$ if $z_{ik}$ is greater than its respective threshold and $\alpha_{ik}=0$ otherwise. For the independent setting, we set $\rho=0,$ while for the dependent settings, we consider both a low correlation with $\rho=0.25$ and a high correlation with $\rho=0.75.$ 
For the tuning of hyper-parameters, we take the candidate sets as $\lambda\in\{0.003,0.004, ..., 0.015\}$ and $\gamma_{0}\in\{0.5, 1, ... , 5.5\}$, and perform $5$-fold CV to select the best estimated $Q$-matrix.
For each setting, 100 repetitions are simulated.
The batch size and the number of epochs are fixed at $50$ and $300$ respectively.


\subsection{Simulation Study 1. DINA Model}
\label{simu-dina}
For the DINA test items, we consider two uncertainty levels, 
$g_{j}=s_{j}=0.1$ or $g_{j}=s_{j}=0.2$ for all $j=1,...,J$.
Figure \ref{fig:DINA_MSE_time_k25} plots the mean batch errors against time for the independent case with $K=5$ (the first row) and $K=25$ (the second row) across different sample sizes and different noise levels.
When $K=5$, we can see that the CD-$1$ algorithm converge well after $6$ seconds for all different sample sizes under different noise levels. 
This suggests that with a small number of latent attributes, the sample sizes and the uncertainty levels do not affect the convergence speed a lot.
Focusing on the second row of Figure \ref{fig:DINA_MSE_time_k25}, we note that although the size of the $Q$-matrix increases from $75$ ($K=5$) to $1875$ ($K=25$), the convergence time only increases by around $10$ seconds, and the CD-$1$ algorithm converges well after just $15$ seconds even when $K=25$. 
This indicates that the proposed method is scalable with the size of the $Q$-matrix. Dependent settings have similar convergence rates and hence the results are omitted.

Figure \ref{fig:DINA_Error_VS_Q} and \ref{fig:DINA_Error_VS_Q_dep} plot different estimation errors against the sizes of the $Q$-matrix for independent and dependent settings respectively. 
For the independent case, in Figure \ref{fig:DINA_Error_VS_Q}, we can see that the OE stays below $16\%$ across all the settings. 
There is a decreasing trend in the OE as the $Q$-matrix size increases due to the increasing sparsity of the true underlying $Q$-matrix.
Our proposed method performs significantly better than the baseline method predicting all the entries of the $Q$-matrix to be 0 (which would produce OE of $36\%$ for $K=5$).
Furthermore, we note that increasing uncertainty level will deteriorate the OTP, making the estimation of positive entries harder.
Increasing the sample size $N$ would in general help improve the estimation accuracy.
For the dependent case, in Figure \ref{fig:DINA_Error_VS_Q_dep}, we observe that the results in the low correlation setting are very similar to that of the independent setting. This suggests that our proposed method is robust when moderate correlations amongst latent attributes exist.
On the other hand, when the correlations amongst the attributes are high, we see increments in all the three error metrics, OE, OTP and OTN. The correlations amongst the attributes would compound the difficulty in estimation of the $Q$-matrix.
However, all the OE's still stay well below 20\%. Hence, our proposed method can still achieve effective learning of the $Q$-matrix when the correlations amongst the attributes are high.

\begin{figure}[H]
  \centering
  \begin{subfigure}{\linewidth}
    \centering
    \includegraphics[scale=0.45]{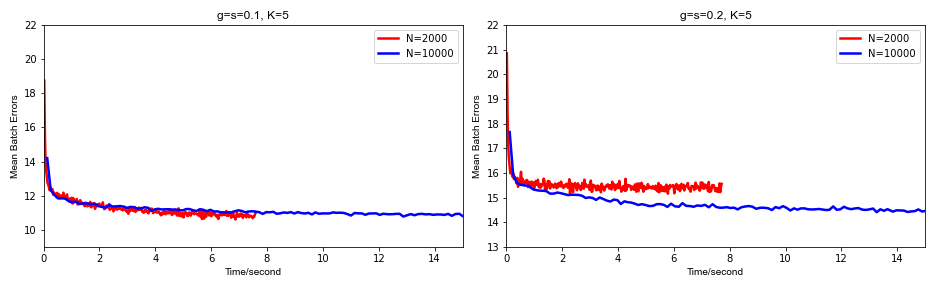}
  \end{subfigure}

  \begin{subfigure}{\linewidth}
    \centering
   \includegraphics[scale=0.45]{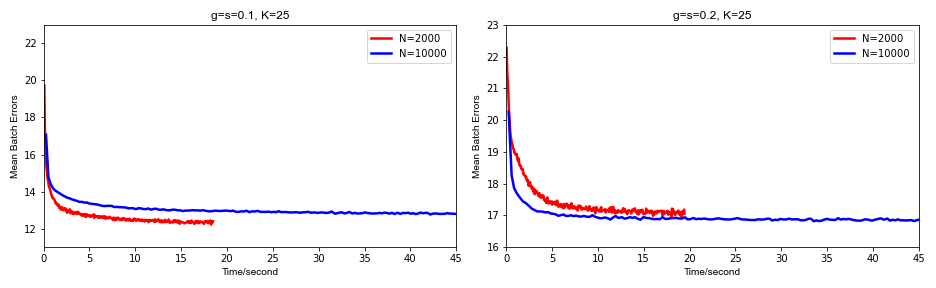}
  \end{subfigure}  
  \caption{Plots of mean batch errors against time for the DINA data.}  
  \label{fig:DINA_MSE_time_k25}
\end{figure}  

\begin{figure}[H]
  \centering
  \begin{subfigure}{\linewidth}
    \centering
    \includegraphics[scale=0.45]{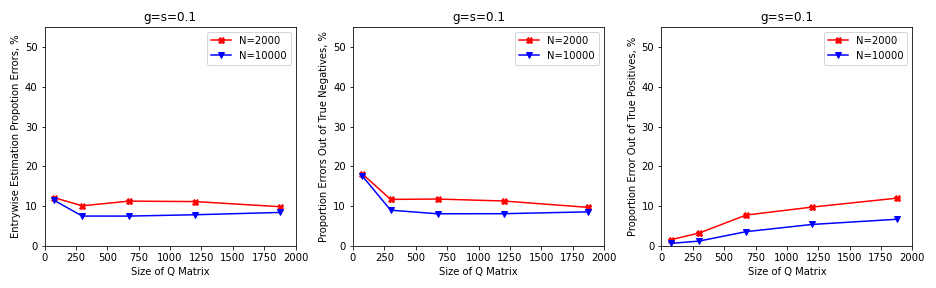}
  \end{subfigure}

  \begin{subfigure}{\linewidth}
    \centering
   \includegraphics[scale=0.45]{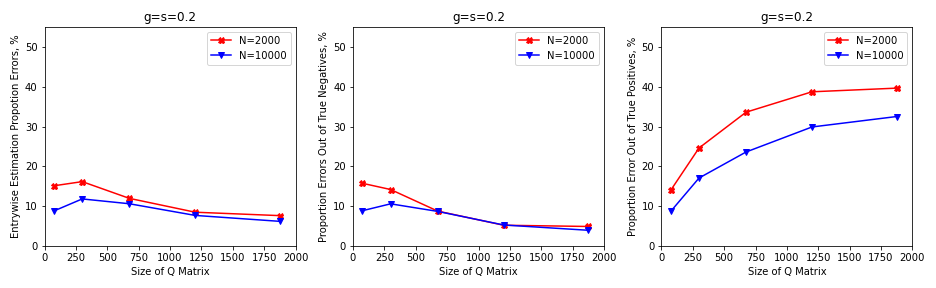}
  \end{subfigure}  
  \caption{Plots of different performance metrics against the size of the $Q$-matrix for the DINA data (independent case).}  
  \label{fig:DINA_Error_VS_Q}
\end{figure}

\begin{figure}[H]
  \centering
  \begin{subfigure}{\linewidth}
    \centering
    \includegraphics[scale=0.45]{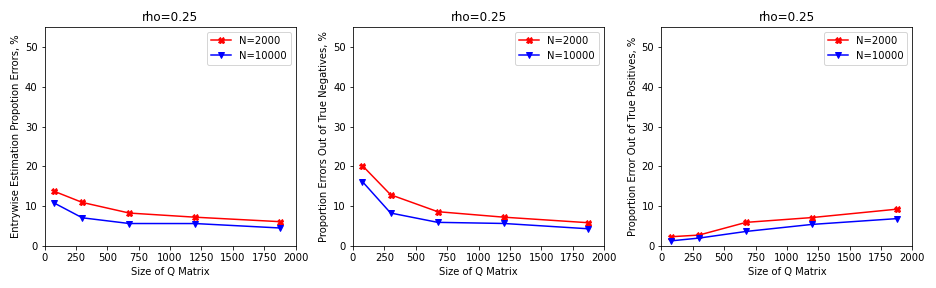}
  \end{subfigure}
  
  \begin{subfigure}{\linewidth}
    \centering
    \includegraphics[scale=0.45]{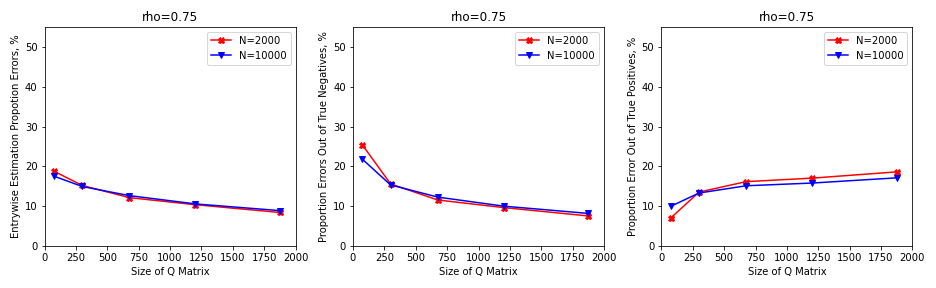}
  \end{subfigure}
  \caption{Plots of different performance metrics against the size of the $Q$-matrix for the DINA data (dependent case with $g=s=0.1$). Row 1 and 2 correspond to correlation settings 0.25 and 0.75 respectively. }  
  \label{fig:DINA_Error_VS_Q_dep}
\end{figure}  





\subsection{Simulation Study 2. ACDM Model}
\label{simu-acdm}
We conduct similar analysis using data generated from the ACDM to examine the convergence speed and estimation accuracy of our proposed method. 
Define $K_{j}^{*}$ to be the number of required attributes for the item $j$. Without loss of generality, we let the first $K_{j}^{*}$ attributes be the required attributes for item $j$, i.e., the corresponding row in the $Q$-matrix is $\bm{q}_{j}= (1,...,1,0,...,0)$ with the first $K_{j}^{*}$ entries being $1$ and all the remaining $K-K_{j}^{*}$ entries being $0$. 
For an ACDM with the identity link function \ref{eq-acdm}, we have $P(R_j=1 \mid \bm{1}_K)=\delta_{j,0}+\sum_{k=1}^{K_{j}^{*}}\delta_{j,k} := p_{j}$, the highest success probability achievable for the most capable subjects. 
Similar to the DINA setting, two different uncertainty levels are considered: case 1. $\delta_{j,0}=0.1$, $p_j=0.9$ for all $j=1,...,J$ and  case 2. $\delta_{j,0}=0.2$, $p_j=0.8$ for all $j=1,...,J$. 
For $k=1,..., K_{j}^{*}$, $\delta_{j,k}$ is set to be $(p_{j}-\delta_{j,0})/K_{j}^{*}$, that is, the contribution of each required attribute to the success probability is equal. 

Figure \ref{fig:ACDM_error_vs_time} shows the convergence speed of our proposed method under the independent setting. 
We observe similar patterns as in the DINA case: 
uncertainty levels and samples sizes do not have significant impacts on the convergence speed. 
Our proposed algorithm is scalable with the size of the $Q$-matrix in the ACDM setting. 
Figure \ref{fig:ACDMerror_vs_Q} and \ref{fig:ACDMerror_vs_Q_dep} plot different estimation metrics against the size of the $Q$-matrix for independent and dependent settings respectively. 
From Figure \ref{fig:ACDMerror_vs_Q}, we can see that the results are very similar to those observed in the DINA model setting,  
which demonstrates that our proposed methods is effective in the ACDM data. 
Furthermore, for the dependent setting in Figure \ref{fig:ACDMerror_vs_Q_dep}, we observe that when the correlation is of 0.25, the estimation accuracy remains similar to that in the independent settings. 
When the correlation is of 0.75, unlike in the DINA setting, the OE, OTP and OTN only increase very slightly. 
In particular, the OE stays well below 16.5\% when $K=5,10,...,25$.
This suggests that when the true data generating model is the ACDM, our proposed method is robust when the correlations amongst the attributes are high.


%

\begin{figure}[H]
  \centering
  \begin{subfigure}{\linewidth}
    \centering
    \includegraphics[scale=0.45]{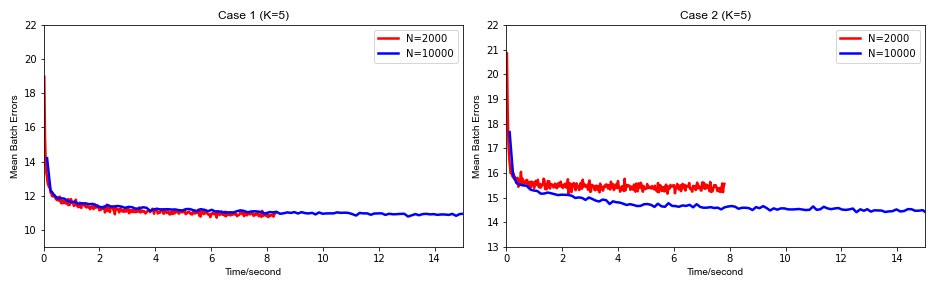}
  \end{subfigure}

  \begin{subfigure}{\linewidth}
    \centering
   \includegraphics[scale=0.45]{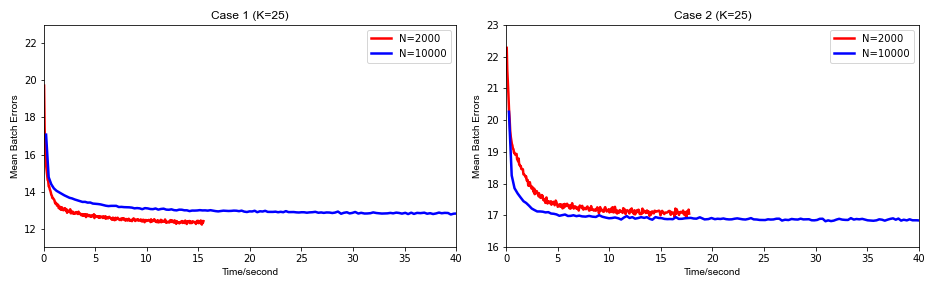}
  \end{subfigure}  
  \caption{Plots of mean batch errors against the time for the ACDM data.  Case 1 represents the setting when $\delta_{j,0}=0.1$, $p_j=0.9$ for all $j=1,...,J$. Case 2 represents the setting when $\delta_{j,0}=0.2$, $p_j=0.8$ for all $j=1,...,J$.}
\label{fig:ACDM_error_vs_time}
\end{figure}

\begin{figure}[H]
  \centering
  \begin{subfigure}{\linewidth}
    \centering
    \includegraphics[scale=0.45]{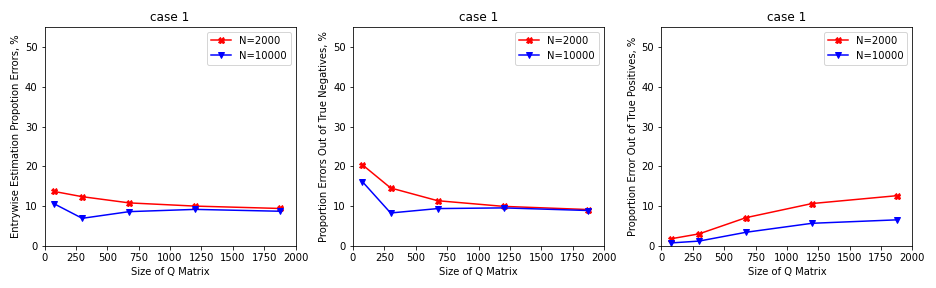}
  \end{subfigure}

  \begin{subfigure}{\linewidth}
    \centering
   \includegraphics[scale=0.45]{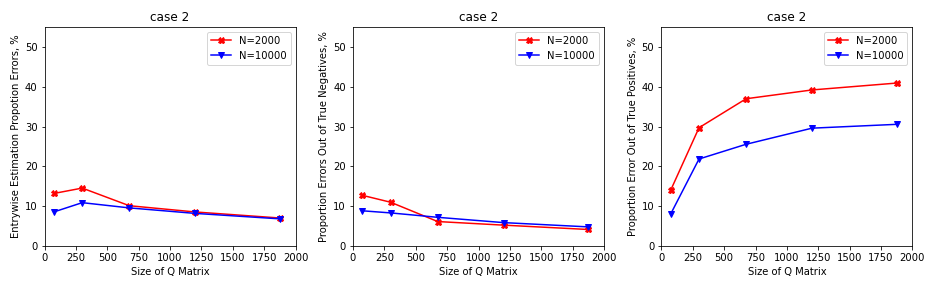}
  \end{subfigure}  
  \caption{Plots of different performance metrics against the size of the $Q$-matrix for the ACDM data (independent case).  Case 1 represents the setting when $\delta_{j,0}=0.1$, $p_j=0.9$ for all $j=1,...,J$. Case 2 represents the setting when $\delta_{j,0}=0.2$, $p_j=0.8$ for all $j=1,...,J$. }
\label{fig:ACDMerror_vs_Q}
\end{figure}

\begin{figure}[H]
  \centering
  \begin{subfigure}{\linewidth}
    \centering
    \includegraphics[scale=0.45]{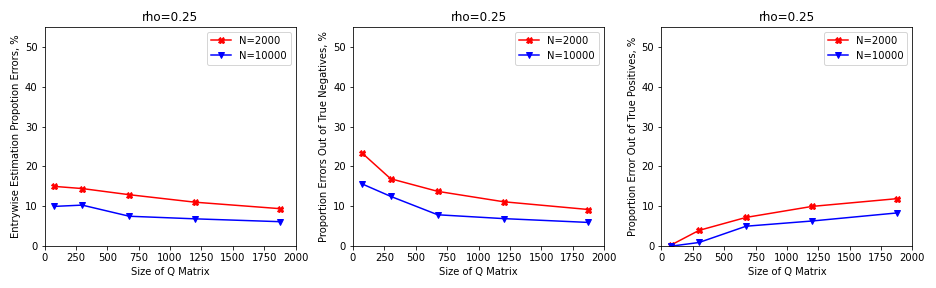}
  \end{subfigure}
  
  \begin{subfigure}{\linewidth}
    \centering
    \includegraphics[scale=0.45]{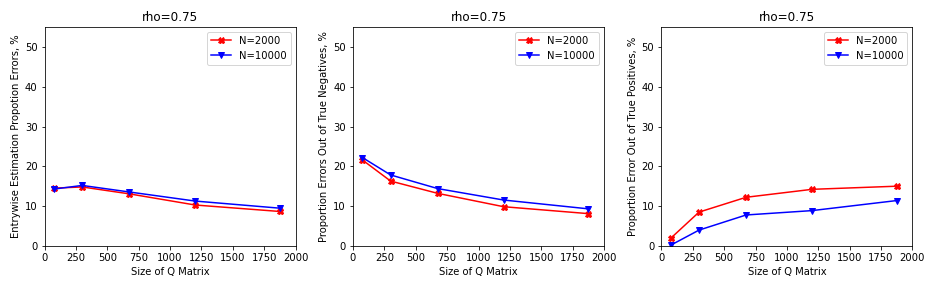}
  \end{subfigure}
  \caption{Plots of different performance metrics against the size of the $Q$-matrix for the ACDM data (dependent case with  $\delta_{j,0}=0.1$, $p_j=0.9$ for all $j=1,...,J$). Row 1 and 2 correspond to correlation settings 0.25 and 0.75 respectively.}  
  \label{fig:ACDMerror_vs_Q_dep}
\end{figure}



\subsection{Simulation Study 3. GDINA Model}
\label{simu-gdina}
Let the highest success probability achievable for the most capable subjects be $P(R_{j}=1\mid \bm{1}_K):=p_{j}$ from Equation (\ref{eq:GDina}). 
Similar to the ACDM setting, we consider two uncertainty levels: case 1. $\delta_{j,0}=0.1$, $p_{j}=0.9$ for all $j=1,...,J$ and case 2. $\delta_{j,0}=0.2$, $p_{j}=0.8$ for all $j=1,...,J$. 
Using the $Q$-matrix specified at the beginning of this section, for each item $j$, we may have $K_j^*=1,2$ or 3. 
When $K_{j}^{*}=1$,we set $\delta_{j,k}=p_{j}-\delta_{j,0}$.
When $K_{j}^{*}=2$, we let $\delta_{j,k}=\delta_{jkk'}=(p_{j}-\delta_{j,0})/3$ and when $K_{j}^{*}=3$ we set $\delta_{j,k}=\delta_{jkk'}=\delta_{jkk'k''}=(p_{j}-\delta_{j,0})/7$. 
As such, the main effects and the interaction terms are all assumed to have the same contributions to the probability of a positive response.  
Both independent and dependent settings are considered. 

Convergence rates under independent setting are summarized in Figure \ref{fig:GDinaerror_vs_time}. 
Similar patterns to the DINA and the ACDM settings can be observed, indicating that our algorithm is scalable to the size of the $Q$-matrix in the GDINA model. 
As before, dependent settings have similar convergence patterns, and hence the results are not presented here.
Behaviors of different estimation metrics over the size of the $Q$-matrix for both the independent and dependent settings are summarized in Figure \ref{fig:GDinaerror_vs_Q} and \ref{fig:GDinaerror_vs_Q_dep} respectively. 

For the independent setting in Figure \ref{fig:GDinaerror_vs_Q}, slightly better estimation accuracy can be observed than in the DINA and the ACDM settings. 
This suggests our proposed methods is effective in the learning the $Q$-matrix from data generated using the GDINA model. One thing to emphasize is that our method is competitive amongst the existing algorithms in the literature. For example, comparing to a similar simulation study in \cite{XuEstimateQmatrix} for $K=5$ independent attributes and $N=2000$, our overall estimation accuracy of around $87\%$ is significantly better than theirs, whose overall accuracy is $71.2\%$. Moreover, our method also has much smaller computational cost than their method. 
For the dependent setting in Figure \ref{fig:GDinaerror_vs_Q_dep}, we observe that the estimation accuracy remains similar to the independent setting when the correlation is of 0.25. When the correlations are increased to 0.75, all the three error metrics only increase very slightly. This observation is similar to the ACDM setting. The OE's remain well below 16.5\% for all $K=5,10,...,25.$ This suggests that when the true data generating model is the GDINA model, the proposed method is fairly robust to high attribute correlations.


 %
\begin{figure}[H]
  \centering
  \begin{subfigure}{\linewidth}
    \centering
    \includegraphics[scale=0.45]{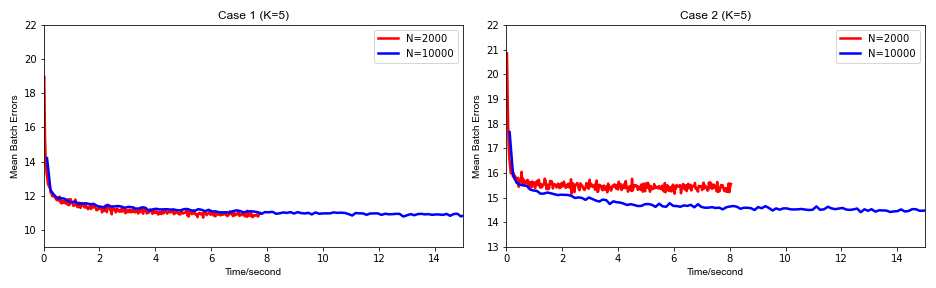}
  \end{subfigure}

  \begin{subfigure}{\linewidth}
    \centering
  \includegraphics[scale=0.45]{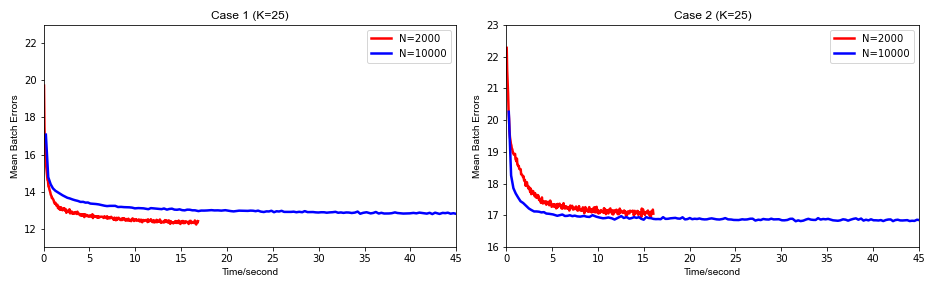}
  \end{subfigure}  
  \caption{Plots of mean batch errors against the size of the $Q$-matrix for the GDINA data. Case 1 represents the setting when $\delta_{j,0}=0.1$, $p_{j}=0.9$ for all $j=1,...,J$. Case 2 represents the setting with higher uncertainty levels when $\delta_{j,0}=0.2$, $p_{j}=0.8$ for all $j=1,...,J$. }
\label{fig:GDinaerror_vs_time}
\end{figure}

\begin{figure}[H]
  \centering
  \begin{subfigure}{\linewidth}
    \centering
    \includegraphics[scale=0.45]{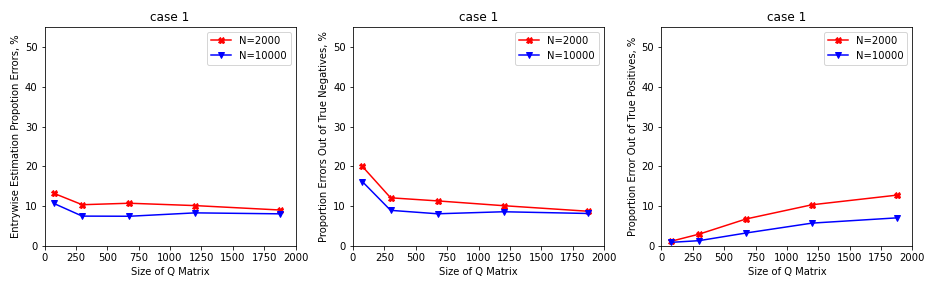}
  \end{subfigure}

  \begin{subfigure}{\linewidth}
    \centering
  \includegraphics[scale=0.45]{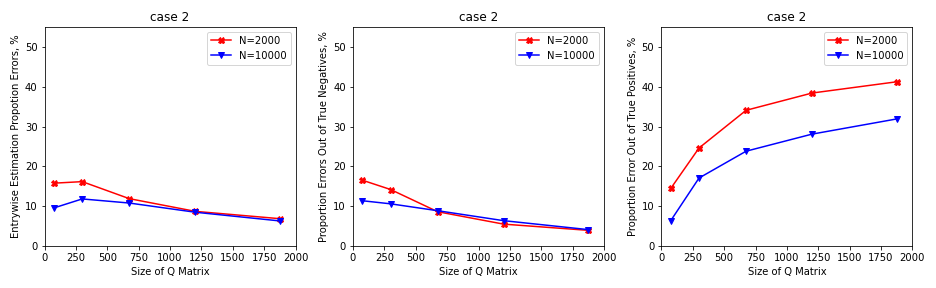}
  \end{subfigure}  
 \caption{Plots of different performance metrics against the size of the $Q$-matrix for the GDINA data (independent case). Case 1 represents the setting when $\delta_{j,0}=0.1$, $p_{j}=0.9$ for all $j=1,...,J$. Case 2 represents the setting with higher uncertainty levels when $\delta_{j,0}=0.2$, $p_{j}=0.8$ for all $j=1,...,J$. }
\label{fig:GDinaerror_vs_Q}
\end{figure}

\begin{figure}[H]
  \centering
  \begin{subfigure}{\linewidth}
    \centering
    \includegraphics[scale=0.45]{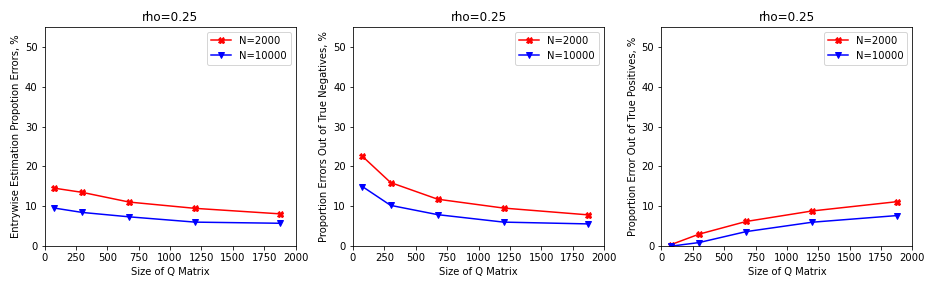}
  \end{subfigure}
  
  \begin{subfigure}{\linewidth}
    \centering
    \includegraphics[scale=0.45]{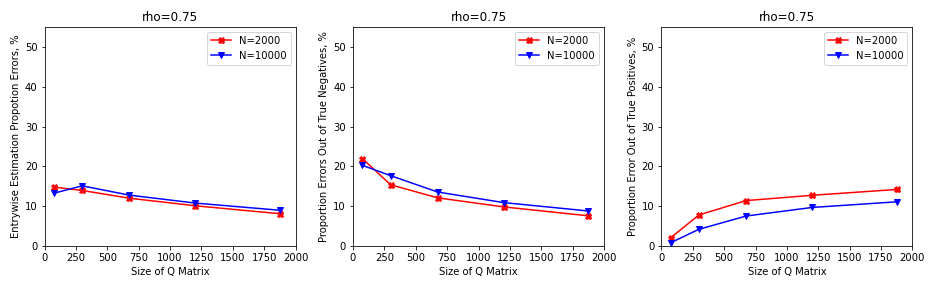}
  \end{subfigure}
  \caption{Plots of different performance metrics against the size of the $Q$-matrix for the GDINA data (dependent case with $\delta_{j,0}=0.1$, $p_{j}=0.9$ for all $j=1,...,J$). Row 1 and 2 correspond to correlation settings 0.25 and 0.75 respectively.}  
 \label{fig:GDinaerror_vs_Q_dep}
\end{figure}


\subsection{Attribute Classifications}\label{sec-implications-acc}

As discussed in Section \ref{sec-robust}, the marginal distributions of the attributes are mis-specified in RBMs, in which a conditional independent structure is assumed. However, in practice, the latent attributes are often highly correlated and the conditional independence assumption may not hold. 
This mis-specification in latent attribute distributions is expected to bring in additional errors in the estimated $Q$-matrix.
In order to understand the practical implications of the mis-specification in the estimated $Q$-matrix, we compare the commonly used attribute classification accuracy (ACC) rate obtained using the estimated $Q$-matrix ($\hat Q$) and the true $Q$-matrix ($Q$).
In particular, when there are $N$ examinees, the ACC of the $k$'th attribute is defined as $$ACC(k):=\frac{1}{N} \sum_{i=1}^N \vert \hat{\alpha}_{ik}-\alpha_{ik}\vert,$$
where $\hat{\alpha}_{ik}$ and $\alpha_{ik}$ represent the estimated and the true attribute values, respectively.

The simulation set-ups remain the same as the dependent settings in Section \ref{sec-simulation-studies}. All the DINA data, the ACDM data and the GDINA data are considered. 
Attribute classifications are performed using the estimated $\hat Q$ and the true $Q$ under the corresponding true underlying CDMs. 
The results are summarized in Table \ref{table: compare-ACC-RBM-CDM}.

Not surprisingly, we observe that the ACC rates obtained using $\hat Q$ are worse than that using $Q$ in all settings across all models.
The errors in $\hat Q$ stem from two sources, the mis-specification error in the latent attributes' marginal distribution and the sample estimation error.
On the other hand, we also note that the ACC rates obtained using $\hat Q$ do not deteriorate too much from using the true $Q$ when sample size is large, especially under the ACDM and GDINA models.
This suggests the $Q$-matrix estimation accuracy in the ACDM and GDINA models may be less prone to the mis-specification in the latent attributes' marginal distribution.
Furthermore, the ACC rates drop as the number of attributes increases in the model. This reflects the increasing difficulty in attribute classifications as the number of attributes increments.
Surprisingly, the ACC rates are generally higher when the correlation amongst attributes is higher. This may be because the higher dependency among the attributes results in fewer numbers of possible attribute patterns, making the estimation relatively easier. Not so surprisingly, we also observe that increasing sample size can in general help improve ACC rates. 

We also conduct simulation studies to explore the potential of using the proposed method to perform latent attribute classifications directly. 
The performance of the proposed method in attribute classifications is satisfactory. 
For more details on the additional simulation results, please refer to the Supplementary Materials.

\begin{table}[H]
\begin{center}
\begin{tabular}{c| c |p{1cm} p{1cm} |p{1cm} p{1cm} |p{1cm} p{1cm} |p{1cm} p{1cm}}
 \hline
\multicolumn{1}{c}{} & \multicolumn{1}{c}{}&\multicolumn{4}{|c}{$N=2000$}&\multicolumn{4}{|c}{$N=10000$} 
 \\
 \hline
\multicolumn{1}{c}{}&\multicolumn{1}{c}{}& \multicolumn{2}{|c}{$\rho=0.25$}&\multicolumn{2}{|c}{$\rho=0.75$}&\multicolumn{2}{|c}{$\rho=0.25$} &\multicolumn{2}{|c}{$\rho=0.75$} \\
 \hline
&Model&$\hat Q$&$Q$&$\hat Q$&$Q$&$\hat Q$&$Q$&$\hat Q$&$Q$ \\ 
 \hline
\multirow{3}{4em}{$K=5$}& DINA& 0.806 & 0.944& 0.888& 0.956& 0.830& 0.945& 0.890& 0.957\\
 &ACDM &0.812& 0.926& 0.911& 0.946& 0.921& 0.928& 0.915& 0.948\\
 &GDINA&0.918&0.928&  0.935&0.949&  0.928& 0.929& 0.947&0.950\\
 \hline
 \multirow{3}{4em}{$K=10$}& DINA&0.801& 0.939& 0.898& 0.954& 0.811& 0.940& 0.894& 0.956\\
 &ACDM &0.815& 0.922& 0.913& 0.946& 0.910& 0.925& 0.906& 0.950\\
 &GDINA&0.885&0.924& 0.939& 0.949& 0.926& 0.926& 0.899& 0.951\\
 \hline
\end{tabular}
\end{center}
\caption{Average ACC rates out of 100 repetitions for $K=5, 10$ attributes respectively obtained using the true CDMs. $\hat Q$ and $Q$ denote the estimated $Q$-matrix from the proposed method and the true $Q$-matrix respectively.}
\label{table: compare-ACC-RBM-CDM}
\end{table}

\section{Real Data Analysis}
\label{sec-real data analysis}
We apply our proposed method to a TIMSS data set. TIMSS provides data on the mathematics and science curricular achievement of the fourth and the eighth grade students across countries such as the U.S. The data set contains  23 mathematical items from TIMSS 2003 items and is packed in the CDM package in R \citep{robitzsch2020package}.
Both a binary scored examinees' response matrix and an associated expert constructed $Q$-matrix are included in the data set. In particular, the binary response matrix consists of 757 observations, and it is therefore of dimension 757 by 23. The $Q$-matrix on the other hand specifies how the 23 items are related to 13 binary mathematical skill attributes, as summarized in Table \ref{table: question_attributes}.

\begin{table}[h]
\small
\begin{tabular}{l|l}
 \hline
 Skill attributes & Items\\
 \hline
1. Understand concepts of a ratio and a unit rate and use language appropriately & 1, 7, 20 \\
\hline
2. Use ratio and rate reasoning to solve real world and mathematical problems &3, 11, 15, 19, 22\\
\hline
3. Compute fluently with multi-digit numbers
and find common factors and multiples& 12, 18\\
\hline
4. Apply and extend previous understandings of
numbers to the system of \\rational numbers&4, 17, 23\\
\hline
5. Apply and extend previous understandings of
arithmetic to algebraic expressions&8, 13, 16, 21\\
\hline
6. Reason about and solve one-variable equations and inequalities & 2, 5, 6,10,14\\
\hline
7. Recognize and represent proportional relationships between quantities& 3, 6\\
\hline
8. Use proportional relationships to solve multi-step ratio and percent problems&11\\
\hline
9. Apply and extend previous understandings of
operations with fractions to \\add, subtract, multiply, and divide rational numbers&4, 8, 18, 23\\
\hline
10. Solve real-life and mathematical problems
using numerical and algebraic \\expressions and
equations&5\\
\hline
11. Compare two fractions with different numerators and different denominators;\\ Understand a
fraction $a/b$ with as $a>1$ a sum of fractions $1/b$&1, 9, 18\\
\hline
12. Solve multi-step word problems posed with
whole numbers and having whole \\number answers using the four operations, including problems in which remainders \\must be interpreted.
Represent these problems using equations with a
letter standing \\for the unknown quantity; Generate a number or shape pattern that follows a given \\rule. Identify apparent features of the pattern that were not explicit in the rule itself&5 , 15\\
\hline
13. Use equivalent fraction as a strategy to add
and subtract fractions&1, 12, 18\\
\hline
\end{tabular}
\caption{Clusters of items according to the underlying skill attributes.}
\label{table: question_attributes}
\end{table}

\begin{figure}[H]
    \centering
    \includegraphics[scale=0.7]{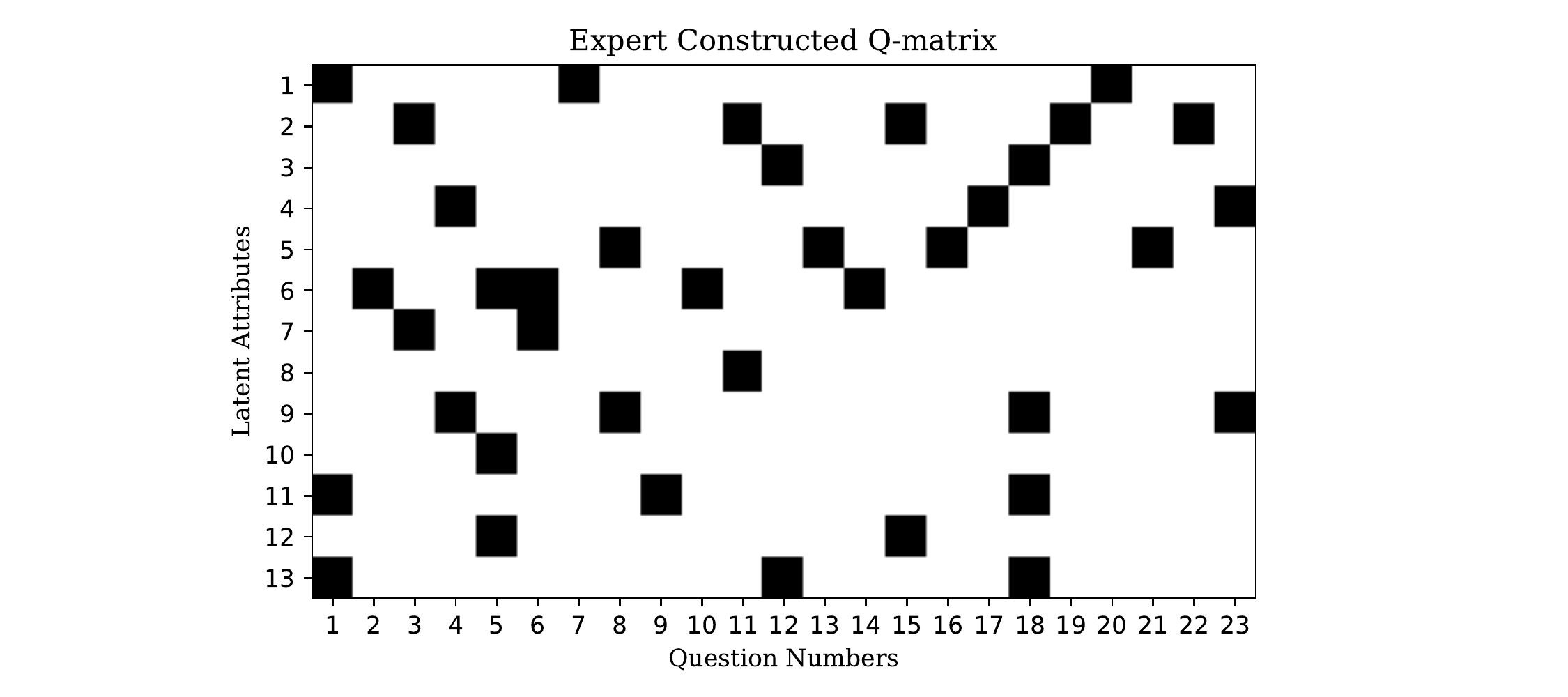}
  \caption{Heat-plot of the expert constructed $Q^0$. The white/black blocks correspond to $q^0_{ij}=0/1$ respectively.}  
 \label{fig:expert-Q}
\end{figure}

Note that the provided $Q$-matrix may not fully represent the ground truth because the construction of the $Q$-matrix by experts is almost always subjective.
In this case, the provided $Q$-matrix was constructed from the consensus of two experts. When they are not able to reach an agreement for any item through discussion, a third expert would step in to resolve the conflict. The percentage of two experts' overall agreement for the constructed $Q$-matrix is only 88.89\%, according to \cite{su2013hierarchical}. 
We denote this expert constructed $Q$-matrix as $Q^0$ and its $(i,j)$th entry as $q^0_{ij}.$
A heat-plot of $Q^0$ is summarized in Figure \ref{fig:expert-Q}. To demonstrate the practical implications of our proposed method, 
we  start with this expert constructed $Q^0$ and explore further whether our proposed method can improve on the quality of the $Q$-matrix to better represent the ground truth.

We initialize the weight matrix with $Q^0$ in our proposed method.
The estimated $Q$-matrix is denoted as $\hat{Q}$ and its $(i,j)$th entry as $\hat{q}_{ij}.$
If we treat the expert constructed $Q^0$ as the truth for evaluation purpose, then 
the entry-wise proportional ``error'' rate, the out of true positives ``error'' rate and out of true negatives ``error'' rate of $\hat{Q}$ are 
0.126, 0.053 and 0.139 respectively. 
The low ``error'' rates suggest $Q^0$ and $\hat{Q}$ are similar and our proposed method can indeed recover the main latent structure, especially the positive entries, in the expert constructed $Q^0$.

\begin{figure}[H]
    \centering
    \includegraphics[scale=0.7]{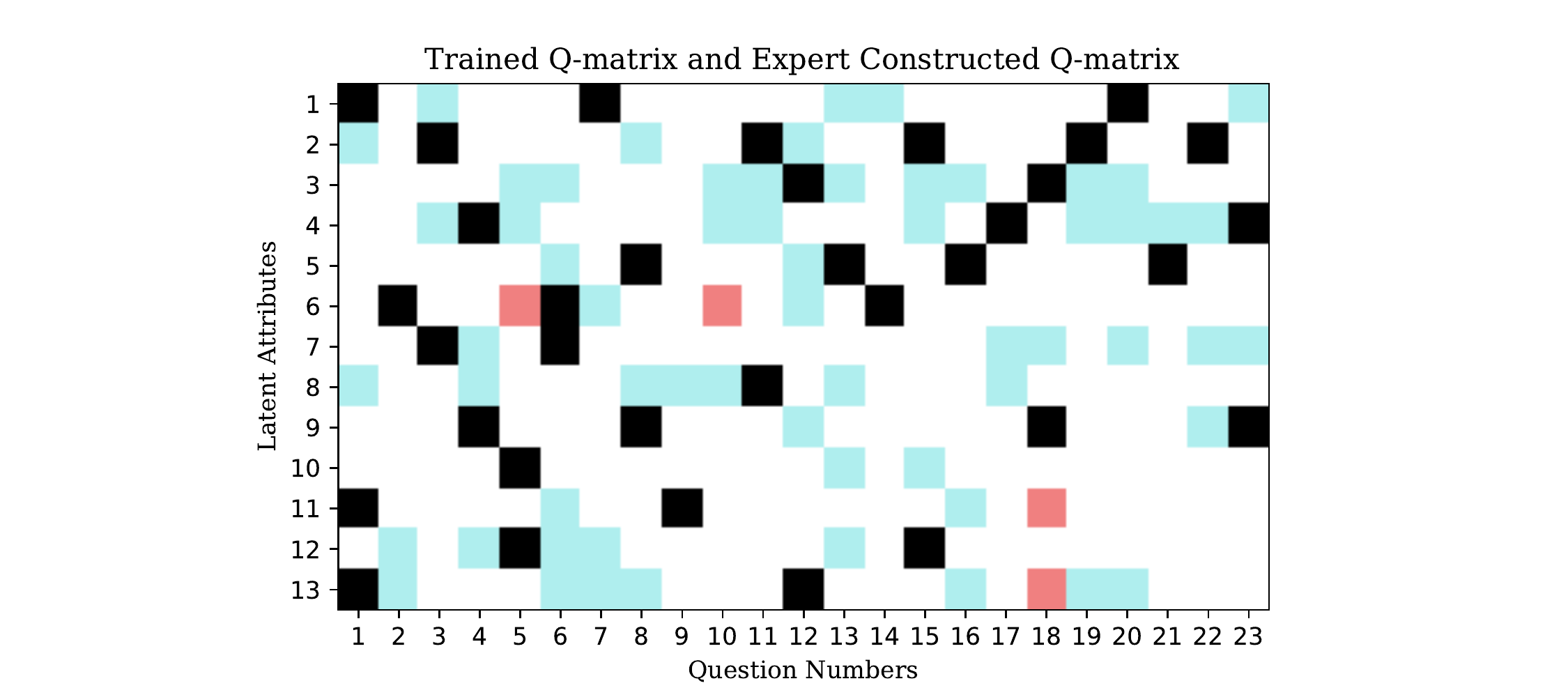}
  \caption{Heat-plot to compare between the estimated $\hat{Q}$ and the expert constructed $Q^0$. 
  The white blocks represent entries $(i,j)$ when both $\hat{q}_{ij}=q^0_{ij}=0.$
  The black blocks represent entries $(i,j)$ when both $\hat{q}_{ij}=q^0_{ij}=1.$
   The red blocks represent entries $(i,j)$ when $\hat{q}_{ij}=0$ and $q^0_{ij}=1.$
  The blue blocks represent entries $(i,j)$ when $\hat{q}_{ij}=1$ and $q^0_{ij}=0.$
 }
 \label{fig:compareQ}
\end{figure} 

Figure \ref{fig:compareQ} presents the heat-plot of the comparison between the estimated $\hat{Q}$ and the expert constructed $Q^0$. 
In particular, white and black entries represent the cases when $\hat{q}_{ij}=q^0_{ij}=0$ and when $\hat{q}_{ij}=q^0_{ij}=1$ respectively.
While blue and red entries represent the cases when $\hat{q}_{ij}=1, q^0_{ij}=0$ and when $\hat{q}_{ij}=0, q^0_{ij}=1$ respectively.
We see that the majority of the positive entries in $Q^0$ are picked up by $\hat{Q}$, and only 4 of them are predicted to be 0 in $\hat{Q}$, as represented by the red blocks in Figure \ref{fig:compareQ}.
This suggests the proposed method can estimate the $Q$-matrix with high sensitivity.
Some of these false negatives do make sense.
For example, item 5 describes three figures arranged in matchsticks with some patterns and asks for the total number of matchsticks that would be used to construct figure 10 if the pattern continues. 
It is a pattern recognition problem and does not seem to be closely related to attribute 6, ``reason about and solve one-variable equations and inequalities''.
 However, we acknowledge that this data driven approach can sometimes make mistakes. For example, the other three false negatives predicted may not make much sense. Take item 10 for example, which reads ``inequality equivalent to $x/3 > 8$''. It clearly requires the knowledge of attribute 6, which is not successfully identified by the proposed method.
On the other hand, the white regions, representing the agreed entry 0's, occupy the majority of the plot. This suggests the specificity is controlled.
Moreover, we see some blue blocks scattering in Figure \ref{fig:compareQ}, which represent the entries that are 0 in $Q^0$ but are predicted to be 1 in $\hat{Q}.$ 
Some of these blocks capture information that is neglected by the expert when constructing the $Q$-matrix. 
Take item 22 for example, whose description is ``At a play, $3/25$ of the people in the audience were children. What percent of audience is this?''
In the expert constructed $Q$-matrix, this item only requires mastering attribute 2. 
However, in our estimated $\hat{Q}$, this item is further related to attribute 4, ``understanding of rational numbers'', 7, ``recognizing proportional relationships'' and 9, ``applying operations with fractions''. 
 Nevertheless, we also want to point out that the proposed method may over-select, resulting in redundant attributes being selected. For example, item 8 reads ``If $x=3,$ what is the value of $-3x$''. The proposed method predicts that it is related to attribute 2, ``Use ratio and rate reasoning to solve real world and mathematical problems'', but in fact item 8 does not seem to be related to attribute 2.
Therefore, careful examination of the predicted entries is still needed, but it can potentially help to improve the quality of the $Q$-matrix.

We further compare the goodness-of-fit of $Q^0$ and $\hat{Q}$ across different CDMs, including the DINA, ACDM and GDINA models using both AIC and BIC as criteria.
We note that out of the three models tested, the ACDM gives the smallest values of both AIC and BIC. Moreover, using $\hat Q$ gives much smaller AIC (19348.71) than using $Q^0$ (19568.17) under the ACDM, which suggests the estimated $\hat Q$ fits better under the ACDM than the expert constructed $Q^0$ in terms of AIC.
On the other hand,
using $\hat Q$ achieves a BIC value of 20313.98, slightly worse than a BIC value of 20286.44 obtained by using $Q^0.$ However, the two values are comparable in size and the improvement is not significant.
Nonetheless, since we do not know what the true underlying model and the $Q$-matrix are, consultation to the domain experts is still needed to make assertive conclusions about which of $\hat Q$ and $Q^0$ is better.

\section{Discussions}
\label{sec-conclusion}
In conclusion, our proposed method using RBMs with $L_1$ penalty can achieve both fast and accurate learning of the large $Q$-matrices in different types of CDMs. 
This is shown by both the theoretical proofs developed in Section \ref{sec-robust} and the simulation studies carried out in Section \ref{sec-simulation-studies}. 
The real data analysis on TIMSS data set further suggests that our method can also work well in real world scenarios, and thus it would provide a powerful tool in large-scale exploratory cognitive diagnosis assessments.

We discuss some potential use cases of our proposed method.
One potential use case is to provide a reasonably accurate $Q$-matrix for cognitive diagnoses such as latent attribute classifications, when no $Q$-matrix or only an inaccurately specified $Q$-matrix is available. 
Depending on the accuracy requirements, the estimated $Q$-matrix can either be used directly in CDMs to perform latent attribute classifications or can serve as a starting point for domain experts for further refinement before use. 
Another potential use case is to provide a $Q$-matrix estimate for test item categorizations and enabling efficient design for future assessments. 
Similarly, whether the estimated $Q$-matrix can be used directly depends on the accuracy requirements in different real settings. 
To add reliability and confidence for direct usage, goodness-of-fit measures such as AIC or BIC can always be evaluated and compared between the estimated $Q$-matrix and the potentially inaccurate specified $Q$-matrix if it is available, as a first step. If the goodness-of-fit of the estimated $Q$-matrix is bad, then either the model used is not appropriate or the estimated $Q$-matrix is inaccurate. In these cases, consultation to domain experts is still necessary.
Nevertheless, our proposed method may help reduce the burden placed on the experts. 
Based on the estimated $Q$-matrix, if one finds out that additional items with specific $q$-vectors need to be included in the test, then it is likely such an item is indeed missing from the original test design. 
In this scenario, we recommend to include the additional item into the test design to keep safe.
Furthermore, in the case when the accuracy requirement is exceptionally high, our proposed method can still help. In this scenario, we recommend to set the penalty term to be 0 and apply CD Algorithm \ref{algo} to train the original RBM on the whole data set to obtain $\hat{\bm{W}}$. Then for each item $j$, experts can rank $\{\vert \hat w_{jk}\vert : k=1,...,K\}$ in a descending order first and pay more attention to those $\hat w_{jk}$ with large absolute values as those correspond to the $q_{jk}$ that are most likely to be 1's.


Note that by initializing the RBM parameters $\bm{W}$, $\bm{b}$ and $\bm{c}$ randomly, the proposed estimation method assumes no prior knowledge of the $Q$-matrix. 
In practice, we may have partial knowledge of the $Q$-matrix, using which we could potentially obtain a better initialization of the parameters. 
For example, we may have a pre-specified $Q$-matrix design with possible mis-specifications in some entries; in such cases, we can initialize the weight matrix $\bm{W}$ and the visible bias vector $\bm{b}$ based on our prior knowledge of the $Q$-matrix.  
Note that  $w_{j, k}$ in  $\bm{W}$ correspond to $\delta_{j,k}q_{j,k}$ in the ACDM.
From the perspective of initialization, we find what affects the learning accuracy most significantly are the signs of the initial values.
So, to keep things simple, we can initialize $\bm{W}$ with the partially available $Q$-matrix directly. For the visible biases, if the underlying model is believed to be the DINA model, by considering $\bm{\alpha}=\bm{0}$, we can derive $b_j=\log(g_j/(1-g_j))$. Under the ACDM or the GDINA model, we can obtain $b_j=\log(\delta_{j,0}/(1-\delta_{j,0}))$ using a similar argument. Though we do not know $g_j$ or $\delta_{j,0}$ in reality, very likely these values are between 0 and 0.5, in which case $b_j<0.$ It is therefore reasonable to initialize each $b_j$ from a Uniform$(-5, 0)$ distribution.
This would help improve the estimation accuracy.

Some limitations of our method include it does not take into account the interactions between the latent attributes due to the assumptions imposed on RBMs. 
In many real world scenarios, it is not uncommon that the latent attributes interact with one another and have joint effects on the distribution of the observed responses. 
One potential way to solve this problem is to apply deep Boltzmann machines (DBMs) to model the distribution of the responses. 
Since DBMs allow interactions between the latent attributes, it will capture the interactions between the latent attributes and take that into account.
Moreover, this paper focus more on the estimation part while inference on the estimated $Q$-matrix is not discussed. It would be interesting to pin down the asymptotic distributional form of this $Q$-matrix estimator to facilitate inferences such as hypothesis testing and constructing confidence intervals.

\section*{Acknowledgments}
The authors are grateful to the Editor-in-Chief Professor Matthias von Davier, an Associate
Editor, and three referees for their valuable comments and suggestions. This research is partially
supported by NSF CAREER SES-1846747, DMS-1712717, and SES-1659328.

\bibliography{bibliography}
\end{document}


\pagenumbering{arabic}
\maketitle

This file contains additional simulation results in Section \ref{sec-appendix-sim} and the proofs of all lemmas and propositions in Section \ref{sec-appendix-proofs}.

\section{Additional Simulation Studies}\label{sec-appendix-sim}
\subsection{Estimating Randomly Sampled Q-Matrix}\label{sec-appendix-randomQ}

In this section, we consider randomly sampled $Q$-matrix in a way that can simulate potentially more challenging scenarios. 
In specific, we include the one-, two- and three-attribute item designs. The exact construction of the $Q$-matrix is as follows. Similar to the construction in the main article, we still fix the dimension of the $Q$-matrix to be $3K$ by $K$, i.e. $3K$ items with $K$ attributes. 
For each row $j$, we first determine which item design it will take by a random sampling scheme. Let $M={K \choose 1}+{K \choose 2}+{K \choose 3}.$ 
The number of required attributes (denoted by $n$) for each item is randomly sampled from $\{1,2,3\}$ with probabilities $\{{K \choose 1}/M, {K \choose 2}/M, {K \choose 3}/M\}$.
Then, $n$ attributes are sampled without replacement from $\{1,2,...,K\}$ with equal probabilities, the corresponding entries in $\mathbf{q}_j$ will be set to 1 and the rest to 0.
Note that this random construction of the $Q$-matrix would somewhat simulate the extreme situations where the easiest learned one-attribute items will be sampled with the smallest probabilities. For example, when $K=15$, the probability to select a one-attribute item is only 0.0261. Furthermore, we also point out that under this random design, there will be a high chance the sampled $Q$-matrix is not identifiable, making the estimation even more difficult.
100 replications for each of $K=5, 10, ..., 25$ are considered and the average results are presented in Figure \ref{fig:randomQ}. For illustration purpose, we only consider the settings when $N=2000$ and when the attributes are independent, for the DINA, the ACDM, and a mixture of the DINA, ACDM, and DINO data. For the data from a mixture of three models, the data are generated from the DINA, ACDM and DINO models with proportions 0.35, 0.35, and 0.3 respectively,  respectively.
All the other set-ups remain the same as the independent settings in Section 4 of the main article.



From Figure \ref{fig:randomQ}, we can observe that the OE's of our proposed method remain controlled for three types of data.
However, we can also see that the OE's worsen and the OTP's become much more volatile compared to the fixed $Q$-matrix design in Section 4 of the main article. This is not surprising because of the increased difficulty in the design where the $Q$-matrices contain more two- and three-attribute items and the number of non-identifiable $Q$-matrices increases significantly.
In line with our observations in the main article, we also observe the increased uncertainty level impact most negatively on the OTP.
However, overall, the proposed method still possesses certain degrees of learning power of the $Q$-matrix even in such extreme situations.

\begin{figure}[H]
  \centering
  \begin{subfigure}{\linewidth}
    \centering
    \includegraphics[scale=0.45]{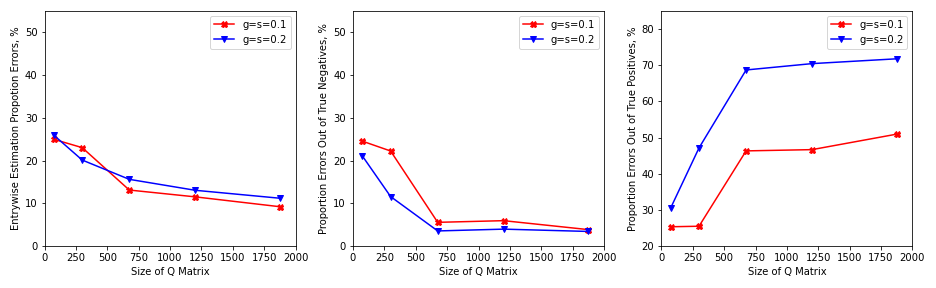}
  \end{subfigure}

  \begin{subfigure}{\linewidth}
    \centering
    \includegraphics[scale=0.45]{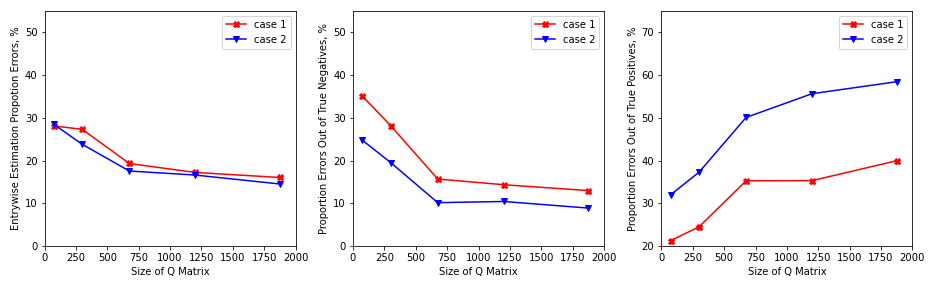}
  \end{subfigure}
  
  \begin{subfigure}{\linewidth}
    \centering
    \includegraphics[scale=0.45]{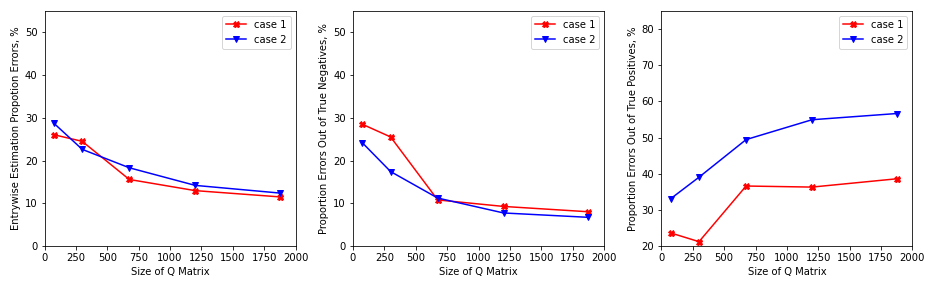}
  \end{subfigure}
  \caption{Plots of different performance metrics against the sizes of the $Q$-matrix. Rows 1 to 3 correspond to the DINA data, the ACDM data and a mixture of the DINA, ACDM and DINO data, respectively. 
  For the DINA and DINO data, two uncertainty levels are represented by $g_j=s_j=0.1$ and $g_j=s_j=0.2$ for all items $j$, where subscripts $j$ are omitted in the legends.
  For both the ACDM data and the GDINA data,  cases 1 and 2 represent the settings when $\delta_{j,0}=0.1$, $p_j=0.9$ and $\delta_{j,0}=0.2$, $p_j=0.8$ for all $j=1,...,J$ respectively.}  
 \label{fig:randomQ}
\end{figure}  

\subsection{Attribute Classifications in Correlated Settings}\label{sec-appendix-ACC}

In this section, we explore the potential of our proposed method in learning the  latent attribute patterns. 
As discussed in the main article, the marginal distributions of the latent attributes are mis-specified in RBMs. Therefore, we would like to explore to what extent our proposed method can perform latent attribute classifications directly when the conditional independence assumption is intensely violated. 
Similarly, ACC rate is used to assess the performance. Recall that the ACC of the $k$'th attribute is defined as $$ACC(k):=\frac{1}{N} \sum_{i=1}^N \vert \hat{\alpha}_{ik}-\alpha_{ik}\vert,$$
where $\hat{\alpha}_{ik}$ and $\alpha_{ik}$ represent the estimated value and the true value respectively.

The simulation set-ups remain the same as the dependent settings in Section \ref{sec-simulation-studies} of the main article. 
The recovered latent attribute matrix corresponding to the optimal estimated $Q$-matrix is returned. All the DINA, ACDM and GDINA data are considered. For each of the 100 replications, the ACC rate for every attribute in each of the settings with $K=5, 10,...,25$ is evaluated. 
The setting-wise average ACC rate is evaluated by computing the average ACC for each attribute out of 100 repetitions first, and then averaging out of all the $K$ latent attributes for each settings of $K=5, 10, ..., 25.$
The results are summarized in Table \ref{table: ACC-all}. 

Overall, we can see that the proposed method performs well in attribute classifications with all ACC rates above 0.85. Furthermore, we also observe that the ACC rates drop as the number of attributes increases in the model.  The attribute patterns would increase as the number of attributes increments, making the estimation more difficult.
Similar to the observations made in the main article, we see the ACC rates are generally higher when the correlations amongst attributes are higher. 
We also point out that increasing sample size can in general improve ACC rates using the proposed method. The performance of the proposed method is better on the ACDM data and the GDINA data than on the DINA data. This is especially obvious when $K$ is relatively small at 5 and 10. This observation is in line with our discussions in Section \ref{sec-robust} of the main article.

\begin{table}[h!]
\small
\centering
\begin{tabular}{p{0.75cm} p{0.95cm} p{1.15cm}|p{0.75cm} p{0.95cm} p{1.15cm}|p{0.75cm} p{0.95cm} p{1.15cm}|p{0.75cm} p{0.95cm} p{1.15cm}}
 \hline
 \multicolumn{6}{c}{$N=2000$}&\multicolumn{6}{|c}{$N=10000$} 
 \\
 \hline
 \multicolumn{3}{c}{$\rho=0.25$}&\multicolumn{3}{|c}{$\rho=0.75$}&\multicolumn{3}{|c}{$\rho=0.25$} &\multicolumn{3}{|c}{$\rho=0.75$} \\
 \hline
DINA&ACDM&GDINA&DINA&ACDM&GDINA&DINA&ACDM&GDINA&DINA&ACDM&GDINA \\ 
 \hline
0.898&0.916&0.916&0.917&0.927&
0.924&0.903&0.916&0.917&0.918&0.932&0.931\\
 \hline
0.897&0.896&0.900&0.888&0.902&0.903&0.901&0.907&0.911&0.885&0.911&0.912\\
\hline
0.878& 0.876& 0.880& 0.880& 0.888& 0.893& 0.891&0.887& 0.893& 0.880& 0.897& 0.900\\
 \hline
0.875& 0.863& 0.869&0.879& 0.885& 0.889& 0.883& 0.879& 0.882&
       0.874& 0.894& 0.893\\
       \hline
0.866& 0.853&0.857& 0.875&0.883&0.887& 0.877& 0.868& 0.874&
       0.874& 0.887& 0.890 \\
 \hline
\end{tabular}
\caption{Average ACC rates for using RBM on the DINA data, the ACDM data and the GDINA data. Rows 1 to 5 correspond to the settings with $K=5, 10, ...,25$ respectively.}
\label{table: ACC-all}
\end{table}

\section{Proofs of Lemmas and Propositions}
\label{sec-appendix-proofs}
Before proving our main propositions 2.1 and 2.2, we first give a lemma which would be used in the proof of the main propositions. 
\begin{lemma}
Assume $\bm{\alpha}$ are independent and $\alpha_k \sim$ Ber$(p_k)$ for $k=1,...,K$. If true model with response R satisfies either the
GDINA model Equation (3) or the DINA model $P(R=1\mid\bm{\alpha})=g+(1-s-g)\alpha_{1}\alpha_{2}...\alpha_{K^{*}}$ for some $s,g$ satisfying $g<1-s$, then the mis-specified linear additive model of $R$ regressed on $(\alpha_{1}, \alpha_{2}, ..., \alpha_{K})$ has the corresponding mean function in the form of $\mathop{\mathbb{E}}^*[R\mid\bm{\alpha}]=\beta_0+\beta_{1}\alpha_{1}+\beta_{2}\alpha_{2}+...+\beta_{K}\alpha_{K}$ with $\beta_{k}=0$ for $k= K^{*}+1,...,K$.
\end{lemma}
\begin{proof}[Proof of Lemma 1]
By the independence assumption and the linear regression theory, we have for $k=1,\ldots, K,$
\begin{align*}
 \beta_{k}&=\frac{1}{Var(\alpha_{k})}Cov\big(\alpha_{k}, R\big)\\ 
 &=\frac{1}{p_{k}(1-p_{k})}Cov\big(\alpha_{k}, R\big). 
\end{align*}
Denote $\alpha_{1,...,K^{*}}:=\{\alpha_1,...,\alpha_{K^{*}}\}$, then by the Law of Total Covariance, we have for $k=K^{*}+1,...,K$,
\begin{align*}
Cov\big(\alpha_{k}, R\big)=\mathop{\mathbb{E}}\big[Cov\big(\alpha_{k},R \mid \alpha_{1,...,K^{*}}\big)\big]+ Cov\big(\mathop{\mathbb{E}}\big[\alpha_{k}\mid\alpha_{1,...,K^{*}}\big], \mathop{\mathbb{E}}\big[R\mid\alpha_{1,...,K^{*}}\big]\big).
\numberthis \label{eq:cov3}
\end{align*}
Applying the independence assumption again, we have
\begin{align*}
Cov\big(\mathop{\mathbb{E}}\big[\alpha_{k}\mid \alpha_{1,...,K^{*}}\big], \mathop{\mathbb{E}}\big[R\mid\alpha_{1,...,K^{*}}\big]\big)&= Cov\big(p_{k}, \mathop{\mathbb{E}}[R \mid\alpha_{1,...,K^{*}}]\big)=0.
\end{align*}
Hence, we only need to consider the first term of (\ref{eq:cov3}). Referring to Figure \ref{fig:conditional_independence}, we know that in both the DINA and the GDINA model setting, $R \perp \!\!\! \perp \alpha_{k} \mid \alpha_{1,...,K^{*}}$ for all $k=K^{*}+1,...,K$.
\begin{align*}
\mathop{\mathbb{E}}\big[Cov\big(\alpha_{k},R\mid\alpha_{1,...,K^{*}}\big)\big]&=0.
\end{align*}
Therefore,
\[
\beta_{k}=\frac{0}{p_{k}(1-p_{k})}=0 \quad \forall k=K^{*}+1,...,K.
\]
\end{proof}

\begin{figure}[H]
\centerline{\includegraphics[scale=0.55]{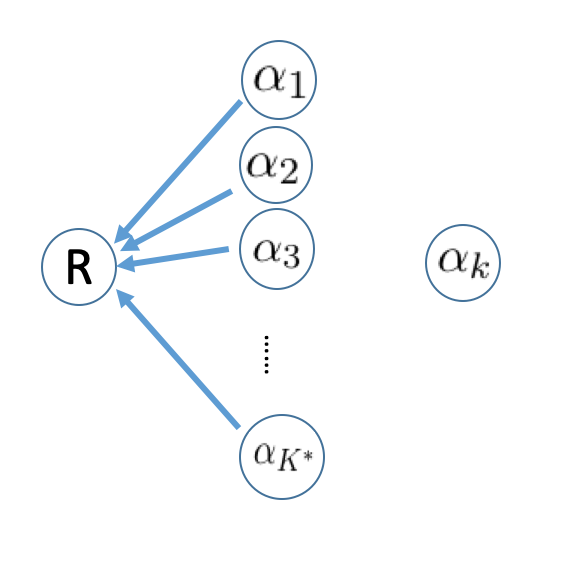}}
\centering
\caption{Illustration of the conditional independence relationship between $R$ and $\alpha_{k}$ given $\alpha_1,...,\alpha_{K^{*}}$ for all $k=K^{*}+1,...,K$}.
\label{fig:conditional_independence}
\end{figure}


Next we give the proofs of our main propositions.

\begin{proof}[Proof of Proposition 1]
First note that by Lemma $1$, we have $\beta_{k}=0$ for $k= K^{*}+1,...,K$.

In the DINA setting, we have
\begin{equation*}
    P(R=1 \mid \bm{\alpha})=
    \begin{cases}
      1-s \quad \text{if}\ \bm{\alpha} \succcurlyeq \mathbf{1}_{K^{*}} \\
      g \quad\text{otherwise},
    \end{cases}
  \end{equation*}
or,
\begin{equation}
    R\mid\bm{\alpha} \sim
    \begin{cases}
      \text{Ber}(1-s) & \text{if}\ \bm{\alpha} \succcurlyeq \mathbf{1}_{K^{*}} \\
      \text{Ber}(g) & \text{otherwise}.
    \end{cases}
  \end{equation}

Under the independence condition, for any $k=1,...,K^{*}$, we have
\begin{align*}
 \beta_{k}&=\frac{1}{Var(\alpha_{k})}Cov\big(\alpha_{k}, R\big) 
  =\frac{1}{p_{k}(1-p_{k})}Cov(\alpha_{k}, R). 
\end{align*}
Consider the following two events which partition the sample space of $\bm{\alpha}$, \\$E_{0,k}:=\big\{\alpha_{1},...,\alpha_{k-1},\alpha_{k+1},...,\alpha_{K^{*}}\mid\prod_{i=1,i\neq k}^{K^{*}}\alpha_{i} =0\big\}$ and $E_{1,k}:=\big\{\alpha_{1},...,\alpha_{k-1},\alpha_{k+1},...,\alpha_{K^{*}}\mid\prod_{i=1,i\neq k}^{K^{*}}\alpha_{i} =1\big\}$. Denote $\alpha_{1,...,K^{*} \setminus k}:=\big\{\alpha_{1},...,\alpha_{k-1},\alpha_{k+1},...,\alpha_{K^{*}}\big\}$.
By the Law of Total Covariance, we have
\begin{align*}
Cov\big(\alpha_{k}, R\big)=\mathop{\mathbb{E}}\big[Cov\big(\alpha_{k},R\mid\alpha_{1,...,K^{*} \setminus k}\big)\big]+ Cov\big(\mathop{\mathbb{E}}\big[\alpha_{k}|\alpha_{1,...,K^{*} \setminus k}\big], \mathop{\mathbb{E}}\big[R\mid\alpha_{1,...,K^{*} \setminus k}\big]\big).
\numberthis \label{eq:cov}
\end{align*}
 Applying the independence condition,
\begin{align*}
Cov\big(\mathop{\mathbb{E}}\big[\alpha_{k}\mid\alpha_{1,...,K^{*} \setminus k}\big], \mathop{\mathbb{E}}\big[R\mid\alpha_{1,...,K^{*} \setminus k}\big]\big)
&= Cov\big(p_{k}, \mathop{\mathbb{E}}\big[R\mid\alpha_{1,...,K^{*} \setminus k}\big]\big)=0.
\end{align*}
Hence, we only need to consider the first term of (\ref{eq:cov}),
\begin{align*}
\mathop{\mathbb{E}}\big[Cov\big(\alpha_{k},R\mid\alpha_{1,...,K^{*} \setminus k }\big)\big]&=\mathop{\mathbb{E}}\big[\mathop{\mathbb{E}}\big[\alpha_{k}R\mid\alpha_{1,...,K^{*} \setminus k}\big]-\mathop{\mathbb{E}}\big[\alpha_{k}\mid\alpha_{1,...,K^{*} \setminus k}\big]\cdot \mathop{\mathbb{E}}\big[R\mid\alpha_{1,...,K^{*} \setminus k}\big]\big].
\numberthis \label{eq:first-term}
\end{align*}
For a fixed $k$, define another two events: $E_{2,k}:=\big\{\bm{\alpha}\mid\alpha_{k}=0\big\}$ and $E_{3,k}:=\big\{\bm{\alpha}\mid\alpha_{k}=1\big\}$. Then in the event of $E_{0,k}$,
\begin{align*}
(\ref{eq:first-term})&= \mathop{\mathbb{E}}\big[\mathop{\mathbb{E}}\big[\alpha_{k}R\mid E_{0,k}\big]-\mathop{\mathbb{E}}\big[\alpha_{k}\mid E_{0}\big] \mathop{\mathbb{E}}\big[R\mid E_{0,k}\big]\big]\\
&=\mathop{\mathbb{E}}\big[\mathop{\mathbb{E}}\big[\alpha_k R\mid E_{0,k}, E_{3,k}\big] P(E_{3,k})+\mathop{\mathbb{E}}\big[\alpha_k R\mid E_{0,k}, E_{2,k}\big] P(E_{2,k})-\mathop{\mathbb{E}}\big[\alpha_{k}] \mathop{\mathbb{E}}\big[R\mid E_{0,k}\big]\big]\\
&=\mathop{\mathbb{E}}\big[g\cdot p_{k}-p_{k} \cdot g\big]\\
&=0.
\end{align*}
In the event of $E_{1,k}$,
\begin{align*}
(\ref{eq:first-term})=& \mathop{\mathbb{E}}\big[\mathop{\mathbb{E}}\big[\alpha_{k}R\mid E_{1,k}\big]-\mathop{\mathbb{E}}\big[\alpha_{k}\mid E_{1,k}\big] \mathop{\mathbb{E}}\big[R\mid E_{1,k}\big]\big]\\
=&\mathop{\mathbb{E}}\big[\mathop{\mathbb{E}}\big[\alpha_k R\mid E_{1,k}, E_{3,k}] P(E_{3,k})+\mathop{\mathbb{E}}\big[\alpha_k R\mid E_{1,k}, E_{2,k}] P(E_{2,k})\\
&-\mathop{\mathbb{E}}\big[\alpha_{k}\big]\cdot \mathop{\mathbb{E}}\big[R\mid E_{1,k}, E_{3,k}\big]\cdot P(E_{3,k})-\mathop{\mathbb{E}}\big[\alpha_{k}\big]\cdot \mathop{\mathbb{E}}\big[R\mid E_{1,k}, E_{2,k}\big]\cdot P(E_{2,k})\big]\\
=&\mathop{\mathbb{E}}\big[(1-s) p_{k}+0-p_{k}(1-s) p_{k}-p_{k} g (1-p_{k})\big]\\
=&p_{k}(1-p_{k})(1-s-g).
\end{align*}
Since the above reasoning works for any $k=1,2,...,K^{*}$, we must have for each $k=1,2,...,K^{*}$,
\begin{align*}
\beta_{k}&=\frac{1}{p_{k}(1-p_{k})}Cov\big(\alpha_{k}, R\big) \\
&=\frac{1}{p_{k}(1-p_{k})}\big(0 \cdot P(E_{0,k})+p_{k}(1-p_{k})(1-s-g)\cdot P(E_{1,k})\big)\\
&=(1-s-g)\prod_{i=1,i\neq k}^{K^{*}}p_{i}\\
&\neq 0.
\end{align*}
\end{proof}


\begin{proof}[Proof of Proposition 2]
Note that by Lemma $1$, we have $\beta_{k}=0$ for $k= K^{*}+1,...,K$.

\noindent Under the independence condition, for any $k=1,...,K^{*}$, we have
\begin{align*}
 \beta_{k}&=\frac{1}{Var(\alpha_{k})}Cov\big(\alpha_{k}, R\big)\\ 
 &=\frac{1}{p_{k}(1-p_{k})}Cov(\alpha_{k}, R). 
 \numberthis \label{eq:beta_k}
\end{align*}
Denote $S:=\big\{1,2,3,...,K^{*}\big\}$. We consider the following 
$2^{K^*}$ events: $E_{0}:=\big\{\bm{\alpha}\mid \alpha_{l}=0, \forall l \in S\}$, $E_{1,i}:= \big\{\bm{\alpha}\mid \alpha_{i}=1, \alpha_{j} = 0, \forall j\neq i \in S\big\}$ for some $i \in S$ (i.e. events that only one of the required variables taking value of $1$ and all others being $0$), $E_{2, (i,j)}:= \big\{\bm{\alpha}\mid \alpha_{i}=\alpha_{j}=1, \alpha_{k} = 0, \forall k\neq i,j \in S\big\}$ for some $i \neq j \in S$ (i.e. events that any two of the required variables are $1$ and all others being $0$), ..., $E_{K^{*}}:=\big\{\bm{\alpha}\mid \alpha_{l}=1, \forall l \in S\big\}$. Note that $E_{0}, E_{1, i}$ for $i \in S$, $E_{2, (i,j)}$ for some $i\neq j \in S$, ..., $E_{K^{*}}$ partition the sample space of $\bm{\alpha}$. The response $R$ would have the following distribution.
\begin{equation}
    R|\bm{\alpha} \sim
    \begin{cases}
      \text{Ber}(\delta_{0}) & \text{if}\ E_{0} \\
       \text{Ber}(\delta_{0}+\delta_{i}) & \text{if}\ E_{1, i}\\
       \text{Ber}(\delta_{0}+\delta_{i}+\delta_{j}+ \delta_{i, j}) & \text{if}\ E_{2, (i,j)}\\
      ...\\
       \text{Ber}\big(\delta_{0}+\sum_{k=1}^{K^{*}}\delta_{k}+...+\delta_{12...K^{*}}\big)&\text{if} \ E_{K^{*}}.
    \end{cases}
  \end{equation}
By the Law of Total Covariance, we have
\begin{align*}
Cov\big(\alpha_{k}, R\big)=\mathop{\mathbb{E}}\big[Cov\big(\alpha_{k},R\mid\alpha_{1,...,K^{*} \setminus k}\big)\big]+ Cov\big(\mathop{\mathbb{E}}\big[\alpha_{k}|\alpha_{1,...,K^{*} \setminus k}\big], \mathop{\mathbb{E}}\big[R\mid\alpha_{1,...,K^{*} \setminus k}\big]\big).
\numberthis \label{eq:cov2}
\end{align*}
Similar to the DINA case, we also have
\begin{align*}
Cov\big(\mathop{\mathbb{E}}\big[\alpha_{k}\mid\alpha_{1,...,K^{*} \setminus k}\big], \mathop{\mathbb{E}}\big[R\mid\alpha_{1,...,K^{*} \setminus k}\big]\big)
&= Cov\big(p_{k}, \mathop{\mathbb{E}}\big[R\mid\alpha_{1,...,K^{*} \setminus k}\big]\big)=0.
\end{align*}
Hence, we only need to consider the first term of (\ref{eq:cov2}),
\begin{align*}
\mathop{\mathbb{E}}\big[Cov\big(\alpha_{k},R\mid\alpha_{1,...,K^{*} \setminus k }\big)\big]&=\mathop{\mathbb{E}}\big[\mathop{\mathbb{E}}\big[\alpha_{k}R\mid\alpha_{1,...,K^{*} \setminus k}\big]-\mathop{\mathbb{E}}\big[\alpha_{k}\mid\alpha_{1,...,K^{*} \setminus k}\big]\cdot \mathop{\mathbb{E}}\big[R\mid\alpha_{1,...,K^{*} \setminus k}\big]\big].
\numberthis \label{eq:first-term2}
\end{align*}
Fix a $k \in S$. Let $S':=\big\{1,2,...,k-1,k+1,...,K^{*}\big\}$. We can define new 
$2^{K^* - 1}$ events: $E_{0}^{*}:=\big\{\alpha_{1,...,K^{*} \setminus k}\mid \alpha_{l}=0 \quad \forall l \in S'\}$, $E_{1, i}^{*}:=\big\{\alpha_{1,...,K^{*} \setminus k}\mid \alpha_{i}=1, \alpha_{l}=0, \forall l \neq i \in S'\big\}$ for some $i \in S'$, $E_{2, (i,j)}^{*}:=\big\{\alpha_{1,...,K^{*} \setminus k}\mid \alpha_{i}=\alpha_{j}=1, \alpha_{l}=0, \forall l \neq i,j \in S'\big\}$ for some $i\neq j \in S'$,..., $E_{K^{*}-1}^{*}:=\big\{\alpha_{1,...,K^{*} \setminus k}\mid \alpha_{l}=1 \quad \forall l \in S'\big\}$. And define $E_{0}':=\big\{\bm{\alpha}\mid \alpha_{k}=0\big\}$ and $E_{1}':=\big\{\bm{\alpha}\mid\alpha_{k}=1\big\}$.

\noindent In the event of $E_{0}^{*}$,
\begin{align*}
(\ref{eq:first-term2})=&\mathop{\mathbb{E}}\big[\mathop{\mathbb{E}}\big[\alpha_{k}R\mid E_{0}^{*}\big]-\mathop{\mathbb{E}}\big[\alpha_{k}\mid E_{0}^{*}\big] \mathop{\mathbb{E}}\big[R\mid E_{0}^{*}\big]\big]\\
=&\mathop{\mathbb{E}}\big[\mathop{\mathbb{E}}\big[\alpha_{k}R\mid E_{0}^{*}, E_{1}'\big] P(E_{1}')+\mathop{\mathbb{E}}\big[\alpha_{k}R\mid E_{0}^{*}, E_{0}'\big] P(E_{0}')\\
&-\mathop{\mathbb{E}}\big[\alpha_{k}\big]\mathop{\mathbb{E}}\big[R\mid E_{0}^{*},E_{1}'\big] P(E_{1}')-\mathop{\mathbb{E}}\big[\alpha_{k}\big]\mathop{\mathbb{E}}\big[R\mid E_{0}^{*},E_{0}'\big] P(E_{0}')\\
=&\mathop{\mathbb{E}}\Big[(\delta_{0}+\delta_{k})p_{k}+(1-p_{k})\cdot 0 - (\delta_{0}+\delta_{k})p_{k}^{2}-\delta_{0}(1-p_{k})p_{k}\Big]\\
=&p_{k}(1-p_{k})\delta_{k}.
\end{align*}
In the event of $E_{1, i}^{*}$ for some $i \in S'$,
\begin{align*}
(\ref{eq:first-term2})=&\mathop{\mathbb{E}}\big[\mathop{\mathbb{E}}\big[\alpha_{k}R\mid E_{1, i}^{*}\big]-\mathop{\mathbb{E}}\big[\alpha_{k}\mid E_{1, i}^{*}\big] \mathop{\mathbb{E}}\big[R\mid E_{1, i}^{*}\big]\big]\\
=&\mathop{\mathbb{E}}\big[\mathop{\mathbb{E}}\big[\alpha_{k}R\mid E_{1, i}^{*}, E_{1}'\big]P(E_{1}')+\mathop{\mathbb{E}}\big[\alpha_{k}R\mid E_{1, i}^{*}, E_{0}'\big]P(E_{0}')\\
&-\mathop{\mathbb{E}}\big[\alpha_{k}\big]\mathop{\mathbb{E}}\big[R\mid E_{1, i}^{*},E_{1}'\big]P(E_{1}')-\mathop{\mathbb{E}}\big[\alpha_{k}\big]\mathop{\mathbb{E}}\big[R\mid E_{1, i}^{*},E_{0}'\big]P(E_{0}')\big]\\
=&\mathop{\mathbb{E}}\big[(\delta_{0}+\delta_{i}+\delta_{k}+\delta_{ik})p_{k}+(1-p_{k})\cdot 0 - (\delta_{0}+\delta_{i}+\delta_{k}+\delta_{ik})p_{k}^{2}-(\delta_{0}+\delta_{i})(1-p_{k})p_{k}\big]\\
=&p_{k}(1-p_{k})(\delta_{k}+\delta_{ik}).
\end{align*}


In the event of $E_{2, (i,j)}^{*}$ for some $i\neq j \in S'$,
\begin{align*}
(\ref{eq:first-term2})=&\mathop{\mathbb{E}}\big[\mathop{\mathbb{E}}\big[\alpha_{k}R\mid E_{2, (i,j)}^{*}\big]-\mathop{\mathbb{E}}\big[\alpha_{k}\mid E_{2, (i,j)}^{*}\big] \mathop{\mathbb{E}}\big[R\mid E_{2, (i,j)}^{*}\big]\big]\\
=&\mathop{\mathbb{E}}\big[\mathop{\mathbb{E}}\big[\alpha_{k}R\mid E_{2, (i,j)}^{*}, E_{1}'\big]P(E_{1}')+\mathop{\mathbb{E}}\big[\alpha_{k}R\mid E_{2, (i,j)}^{*}, E_{0}']P(E_{0}')\\
&-\mathop{\mathbb{E}}\big[\alpha_{k}\big]\mathop{\mathbb{E}}\big[R\mid E_{2, (i,j)}^{*},E_{1}'\big]P(E_{1}')-\mathop{\mathbb{E}}\big[\alpha_{k}\big]\mathop{\mathbb{E}}\big[R\mid E_{2, (i,j)}^{*},E_{0}'\big]P(E_{0}')\big]\\
=&\mathop{\mathbb{E}}\big[(\delta_{0}+\delta_{i}+\delta_{j}+\delta_{k}+\delta_{ij}+\delta_{ik}+\delta_{jk}+\delta_{ijk})p_{k}+(1-p_{k})\cdot 0 \\
&- (\delta_{0}+\delta_{i}+\delta_{j}+\delta_{k}+\delta_{ij}+\delta_{ik}+\delta_{jk}+\delta_{ijk})p_{k}^{2}-(\delta_{0}+\delta_{i}+\delta_{j}+\delta_{ij})(1-p_{k})p_{k}\big]\\
=&p_{k}(1-p_{k})(\delta_{k}+\delta_{ik}+\delta_{jk}+\delta_{ijk}).
\end{align*}

Continuing this process and substitute the relevant values into Equation (\ref{eq:beta_k}), we can show that 
\begin{equation}
    \beta_{k} =
    \begin{cases}
      \delta_{k} & \text{if}\ E_{0}^{*} \\
      \delta_{k}+\delta_{ik} & \text{if}\ E_{1, i}^{*}\\
       \delta_{k}+\delta_{ik}+\delta_{jk}+\delta_{ijk}& \text{if}\ E_{2, (i,j)}^{*}\\
      ...\\
     \delta_{k}+\sum_{i=1, i \neq k}^{K^{*}}\delta_{ik}+...+\delta_{1...K^{*}}&\text{if} \ E_{K^{*}-1}^{*}.
    \end{cases}
  \end{equation}
  
Since the above holds for all $k=1,2,3,...K^{*}$, we have for each $k=1,2,3,...K^{*}$,
\begin{align*}
\beta_{k}=&\delta_{k} \cdot P(E_{0}^{*}) + \sum_{i \in S'} (\delta_{k}+\delta_{ik})\cdot P(E_{1, i}^{*})+\sum_{i, j \in S', i\neq j}(\delta_{k}+\delta_{ik}+\delta_{jk}+\delta_{ijk})\cdot P(E_{2, (i,j)}^{*})+...\\
&+ \big(\delta_{k}+\sum_{i=1, i \neq k}^{K^{*}}\delta_{ik}+...+\delta_{1...K^{*}}\big)\cdot P(E_{K^{*}-1}^{*}) \numberthis \label{eq:beta_k1}
\end{align*}
Assuming monotonicity in acquiring an additional skill, we can show all the terms in (\ref{eq:beta_k1}) are greater than 0. The first term is positive as both $\delta_{k}$ and $P(E_{0}^{*})$ are positive. To see why the second term is positive, consider two examinees, one with skill set $\bm{\alpha}_{1}=\big\{\bm{\alpha}\mid \alpha_{i}=1, \alpha_{l}=0, \quad\forall l \neq i \in S\big\}$ while the other with skill set $\bm{\alpha}_{2} = \big\{\bm{\alpha}\mid \alpha_{i}=\alpha_{k}=1, \alpha_{l}=0, \quad\forall l \neq i,k \in S\big\}$. Then we know according to Equation (3), $P(R=1\mid \bm{\alpha_{1}})=\delta_0+\delta_{i}$ and $P(R=1\mid \bm{\alpha_{2}})=\delta_0+\delta_{i}+\delta_{k}+\delta_{ik}$. The monotonicity assumption then implies $P(R=1\mid \bm{\alpha_{2}})-P(R=1\mid \bm{\alpha_{1}})=\delta_{k}+\delta_{ik}>0$. Hence the second term is positive. We can use a similar strategy to show all the terms in (\ref{eq:beta_k1}) are positive and thus reach the conclusion that $\beta_{k} \neq 0$ for each $k=1,2,3,...K^{*}$.
\end{proof}

\begin{proof}[Discussion of Remark 2.]
Conditional on $\alpha_1, \alpha_2,...,\alpha_{K^*}$,  consider adding one
$\alpha_k$, for any $k=K^*+1,...,K,$ into the main effect regression model, then its coefficient can be expressed as
\begin{align*}
\beta_k
&=\frac{Cov\Big(R-\mathbb{E}^*[R\mid \alpha_1,...,\alpha_{K^*}],\quad \alpha_k-\mathbb{E}^*[\alpha_k\mid\alpha_1,...,\alpha_{K^*}]\Big)}{Var\Big(R - \mathbb{E}^*[R\mid \alpha_1,...,\alpha_{K^*}]\Big)},
\end{align*}

where $\mathbb{E}^*[A\mid B]$ is the the regression mean function of $A$ on $B$.
In the special case when $K^*=1,$ we seek to show $\beta_k=0.$  
When $K^*=1,$ note that we must have $\mathbb{E}^*[R\mid \alpha_1]=\mathbb{E}[R\mid \alpha_1].$ 
This is because $\alpha_1$ can only take values of 0 or 1. These two variability's can be modeled exhaustively by the free intercept and the only coefficient in the regression mean function. 
Note that when $K^*>1,$ this may not hold in general.
Note by the Law of Total Covariance,
\begin{align*}
&Cov\Big(R-\mathbb{E}^{*}[R\mid \alpha_1],\quad \alpha_k-\mathbb{E}^*[\alpha_k\mid\alpha_1]\Big)\\
=&\mathbb{E}\Big\{Cov\Big(R-\mathbb{E}[R\mid \alpha_1],\quad \alpha_k-\mathbb{E}[\alpha_k\mid \alpha_1]\mid \alpha_1\Big)\Big\}\numberthis\label{eq: residual1}\\
&+Cov\Big\{\mathbb{E}\big(R-\mathbb{E}[R\mid\alpha_1]\mid \alpha_1\big),\quad \mathbb{E}\big(\alpha_k-\mathbb{E}[\alpha_k\mid \alpha_1]\mid \alpha_1\big)\Big\}.\numberthis\label{eq: residual2}
\end{align*}

Note $\eqref{eq: residual2}=0$ and
\begin{align*}
\eqref{eq: residual1}&=\mathbb{E}\Big\{\mathbb{E}\Big[\big(R-\mathbb{E}[R\mid \alpha_1]\big)\big(\alpha_k-\mathbb{E}[\alpha_k\mid \alpha_1]\big)\mid \alpha_1\Big]+\mathbb{E}\Big[\alpha_k-\mathbb{E}[\alpha_k\mid \alpha_1]\mid \alpha_1\Big] \mathbb{E}\Big[\alpha_k-\mathbb{E}[\alpha_k\mid \alpha_1]\mid \alpha_1\Big]\Big\}  \\
&=\mathbb{E}\Big\{\mathbb{E}\Big[(R-\mathbb{E}[R\mid \alpha_1])(\alpha_k-\mathbb{E}[\alpha_k\mid \alpha_1])\mid \alpha_1\Big]\Big\}\\
&=\mathbb{E}\Big\{\mathbb{E}\Big[R\alpha_k-R\mathbb{E}(\alpha_k\mid \alpha_1)-\alpha_k\mathbb{E}(R\mid \alpha_1)+\mathbb{E}(R\mid \alpha_1)\mathbb{E}(\alpha_k\mid \alpha_1)\mid \alpha_1\Big]\Big\}\\
&=\mathbb{E}\Big\{\mathbb{E}[R\alpha_k\mid\alpha_1]-\mathbb{E}[R\alpha_k\mid\alpha_1]-\mathbb{E}[R\alpha_k\mid\alpha_1]+\mathbb{E}[R\alpha_k\mid\alpha_1]\Big\}\\
&=0.
\end{align*}

Where the second line follows from $\mathbb{E}\Big[\alpha_k-\mathbb{E}[\alpha_k\mid \alpha_1]\mid \alpha_1\Big]=0$ and the third line follows from the fact that $\mathbb{E}[R\mid\alpha_1]\mathbb{E}[\alpha_k\mid\alpha_1]=\mathbb{E}[R\alpha_k\mid\alpha_1]$ by the conditional independence between $R$ and $\alpha_k$ given $\alpha_1.$
Therefore, $\beta_k=0$.
\end{proof}